\documentclass[11pt,a4paper]{article}
\pdfoutput=1
\usepackage[utf8]{inputenc}
\usepackage[T1]{fontenc}

\usepackage[dvipsnames,usenames]{color}
\usepackage[colorlinks=true,urlcolor=Blue,citecolor=Green,linkcolor=BrickRed]
{hyperref}
\usepackage[usenames,dvipsnames]{xcolor}
\usepackage{amsthm,amsmath,amssymb, mathabx}
\usepackage{fullpage}
\usepackage[ruled,noend,linesnumbered]{algorithm2e}     
\usepackage{enumerate,comment}
\usepackage{url}
\usepackage[capitalise]{cleveref}
\usepackage{todonotes}
\usepackage{standalone}
\usepackage[noadjust]{cite}
\usepackage{mathtools}
\usepackage{enumerate}
\usepackage{pgfplots}
\usepackage{subcaption}

\newcommand{\Oh}{{O}}

\newcommand{\cost}{\mathsf{cost}}

\newtheorem{theorem}{Theorem}[section]
\newtheorem{property}[theorem]{Property}

\newtheorem{lemma}[theorem]{Lemma}

\theoremstyle{definition}   
\newtheorem{definition}[theorem]{Definition}

\newtheorem{oq}{Open Question}
\newtheorem{challenge}{Challenge}

\def\eps{\varepsilon}

\begin{document}

	\title{Near-Optimal Compression for the Planar Graph Metric\thanks{Supported in part by Israel Science Foundation grant 794/13.}}

\author{Amir Abboud	\thanks{Stanford University, Department of Computer Science, \texttt{abboud@cs.stanford.edu}.} 
		\and Pawe{\l} Gawrychowski \thanks{University of Haifa, Department of Computer Science, \texttt{gawry@mimuw.edu.pl}.} 
		\and Shay Mozes\thanks{Interdisciplinary Center Herzliya, Efi Arazi School of Computer Science, \texttt{smozes@idc.ac.il}.} 
		\and Oren Weimann\thanks{University of Haifa, Department of Computer Science, \texttt{oren@cs.haifa.ac.il}.}} 
		
		\date{}
		
	\maketitle

	\begin{abstract}
The Planar Graph Metric Compression Problem is to compactly encode the distances among $k$ nodes in a planar graph of size $n$. 
Two na\"ive solutions are to store the graph using $O(n)$ bits, or to explicitly store the distance matrix with $O(k^2 \log{n})$ bits.
The only lower bounds are from the seminal work of Gavoille, Peleg, Prennes, and Raz [SODA'01], who rule out compressions into a polynomially smaller number of bits, for \emph{weighted} planar graphs, but leave a large gap for unweighted planar graphs.
For example, when $k=\sqrt{n}$, the upper bound is $O(n)$ and their constructions imply an $\Omega(n^{3/4})$ lower bound.
This gap is directly related to other major open questions in labelling schemes, dynamic algorithms, and compact routing.
 
Our main result is a new compression of the planar graph metric into $\tilde{O}(\min (k^2 , \sqrt{k\cdot n}))$ bits, which is optimal up to log factors.
Our data structure breaks an $\Omega(k^2)$ lower bound of Krauthgamer, Nguyen, and Zondiner [SICOMP'14] for compression using minors, and the lower bound of Gavoille et al. for compression of weighted planar graphs.
This is an unexpected and decisive proof that weights can make planar graphs inherently more complex. 
Moreover, we design a new \emph{Subset Distance Oracle} for planar graphs with $\tilde O(\sqrt{k\cdot n})$ space, and $\tilde O(n^{3/4})$ query time.

Our work carries strong messages to related fields. In particular, the famous $O(n^{1/2})$ vs. $\Omega(n^{1/3})$ gap for distance labelling schemes in planar graphs \emph{cannot} be resolved with the current lower bound techniques.

	\end{abstract}

\newpage

\section{Introduction}

The shortest path metric of planar graphs is one of the most popular and well-studied metrics in Computer Science. 
Countless papers, surveys, and textbooks address the computational challenges that arise when dealing with it.
In this paper, we address a core problem about this metric that has remained poorly understood. We ask: \emph{How compressible is it?}  That is, how many bits do we need, information theoretically, in order to describe a set of distances in a planar graph?

As we discuss shortly, a better understanding of this core question is crucial to making progress on some of the biggest open problems in other well-studied subjects such as Sparsification, Labeling Schemes, Dynamic Algorithms, and Compact Routing Schemes.

First, let us define our problem more formally.
In the \emph{Metric Compression} problem, we are given a set $S$ of $k$ points in some metric space with distance function $d$, such as the metric of distances in an $n$ node planar graph, and the goal is to find an encoding $\mathcal{C}$ that is as short as possible, yet still allows us to compute $d(v_i,v_j)$ for any two points $v_i,v_j \in S$.

\begin{definition}[The Planar Graph Metric Compression Problem]
Given an unweighted, undirected planar graph $G$ on $n$ nodes, and a subset $S$ of $k$ distinguished nodes in $G$, compute a bit string $\mathcal{C}$ that encodes the distances between all pairs of nodes in $S$. That is, there is a decoding function $f$ that given the encoding $\mathcal{C}$ and any two nodes $v_i,v_j \in S$ returns the $v_i$-to-$v_j$ distance in $G$  (i.e., $f(v_i,v_j,\mathcal{C}) = d_G(v_i,v_j)$).
\end{definition}

There are two na\"ive ways to solve this problem. First, we can store all the distances explicitly as a $k \times k$ matrix in the encoding $\mathcal{C}$. The distance in a graph on $n$ nodes is some number in $\{0,1,\ldots,n\}$, and so this matrix can be encoded using $O(k^2 \log{n})$ bits.
The second option, which is better whenever $k^2 > n$, is to let the encoding be the graph $G$ itself. Na\"ively, this is $O(n \log n)$ bits, and more sophisticated encodings give $O(n)$ \cite{Turan84,MunroR97,ChiangLL01,BF10}. 
Using the $\tilde{O}(\cdot)$ notation to hide polylog factors, we get a na\"ive upper bound of $\tilde{O}(\min\{k^2,n\})$ for our problem.
Is this the best possible, or could there be a much better compression into $\tilde{O}(k \cdot n^{\eps})$ or even $\tilde{O}(k)$ bits?

For context, let us look at other metrics. 
One example of a metric that admits an ultra-efficient compression into $\tilde{O}(k)$ bits is the metric of trees or bounded treewidth graphs \cite{CZ00,GPPR04}. 
For most metrics of interest, however, the exact or lossless version of the compression problem is too difficult and no non-trivial upper bounds, beyond log-factor improvements, are possible.
For example, in general (non-planar) graphs there is a simple $\Omega(k^2)$ lower bound: in any compression, each of the $2^{{k \choose 2}}$ possible graphs on $k$ nodes must be encoded differently.
Instead, it is popular to seek the optimal \emph{lossy} compression from which the metric can be recovered \emph{approximately}, e.g. up to a multiplicative $(1+\eps)$ error.
For example, the classical Johnson-Lindenstrauss \cite{JL84,Ach03} embedding allows one to compress a set of $k$ points in Euclidean $d$-dimensional space into roughly $O(k /\eps^2 \cdot \log^2{k})$ bits, so that the distances between the points can be recovered up to a $(1+\eps)$ factor, and a recent breakthrough of Indyk and Wagner \cite{IW17} reduced the bound to roughly $O(k/\eps^2 \cdot \log{k} \cdot \log{(1/\eps}))$ which is tight up to a $\log{(1/\eps)}$ factor.

Indeed, if we are willing to pay a $(1+\eps)$ error, there are ingenious compressions of the planar graph metric into $\tilde{O}(k)$ bits \cite{Thorup04,Klein05,KKS11}.
But do we have to pay this error, or are planar graphs restricted enough to allow for non-trivial compression?

\begin{oq}
\label{oq1}
Can we beat $\tilde{O}( \min\{k^2,n\} )$ bits for planar graph metric compression?
\end{oq}

There are some lower bounds in our way. 
From the seminal work of Gavoille, Peleg, P{\'e}rennes, and Raz \cite{GPPR04} we know that the metric of \emph{weighted} planar graphs, where the edge weights are polynomially bounded, does not admit any non-trivial compression. 
The authors show that any Boolean $k \times k$ matrix can be ``encoded'' using the distances among a set of $2k$ nodes in a weighted planar graph on $n=O(k^2)$ nodes, where the edge weights are in $[k]$. Since we cannot compress an arbitrary $k \times k$ matrix into $o(k^2)$ bits, we get a nearly-tight lower bound of $\Omega(\min\{ k^2,n\})$ for \emph{weighted} planar graphs.
For unweighted planar graphs, Gavoille et al. simply subdivide the edges in their construction and the number of nodes in the encoding grows to $n=\Theta(k^3)$, which leads to a much weaker lower bound of $\Omega(\sqrt{k\cdot n})$ (see Section~\ref{sec:lowerbound} for more details). 
For example, when $k=\sqrt{n}$, the upper bound is $\tilde{O}(n)$ and the lower bound is $\Omega(n^{3/4})$.
This subdivision of edges is rather na\"ive, and the overall lower bound construction does not seem to capture the full power of the planar graph metric. In fact, it can be simulated by a grid graph \cite{AD16}. 
This naturally suggests the following intriguing challenge of finding a more clever encoding of matrices into planar graphs, which would lead to a negative resolution to Open Question~\ref{oq1}.

\begin{challenge}
\label{ch1}
Can we encode an arbitrary $k \times k$ Boolean matrix $M$ using the distances among a subset of $2k$ nodes $\{v_1,\ldots,v_{2k}\}$ in an \emph{unweighted planar graph} with $O(k^2)$ vertices, so that we can determine $M[i,j]$ by only looking at the distance between $v_i$ and $v_{k+j}$ in our graph?
\end{challenge}

Before presenting our results, let us discuss the state of the art on questions that are closely related to ours, in which we are interested in data structures that are not only as succinct as possible, but also have other desirable features.
Along the way, we give further reasons to be pessimistic about the possibility of a non-trivial compression. 

\paragraph{Sparsification.}
A natural way to compress a graph is by deleting or contracting some of its edges and nodes. 
Finding small subgraphs or minors that preserve or approximate the distances among a given subset of $k$ nodes have been studied for planar graphs \cite{Gupta01,CXKR06,EEST08,BG08,EGK+14,KNZ14,CGH16,GR16,GHP17,KR17arxiv} and for general graphs \cite{BCE03,CE06,Woo06,Pettie09,CGK13,KV13,Parter14,KKN15,Kavitha15,BV16,AB16so,AB16st,Bodwin17}.
Such compressions are appealing algorithmically, since we can readily feed them into our usual graph algorithms, and recent research suggests that, in many settings, near-optimal compression bounds can be achieved using such sparsifiers (e.g. when compressing general graphs with additive error \cite{AB16st,ABP17}).
A discouraging lower bound of Krauthgamer, Nguyen, and Zondiner \cite{KNZ14} shows that even in the case of unweighted grid graphs, it is impossible to beat the na\"ive bound using a (possibly weighted) minor.
Thus, a positive answer to Open Question~\ref{oq1} will have to involve a more complicated data structure.

\paragraph{Labelling Schemes.}
An appealing way to represent graphs is to assign a label $\ell_v$ to each node $v$, so that by looking at the labels of two nodes $\ell_s,\ell_t$ we can infer certain properties such as the distance between them $d(s,t)$.
Finding so-called \emph{distance labelling schemes} in which the labels are as short as possible is a classical subject of study \cite{GrahamPollak,Kannan,GPPR04,Peleg2000proximity}.
Such labels are used for efficient algorithms both in theory \cite{thorup2001compact,AbrahamCG12} and practice \cite{hublabelling}.
An inspiring lecture by Stephen Alstrup at HALG 2016 surveys breakthroughs \cite{alstrup2015optimal,alstrup2015adjacency,alstrup2015distance} achieved in this field in the last few years, all of which involve shaving constants or logarithmic factors.
A famous open question is to close the rare \emph{polynomial} gap in the bounds for planar graphs that has been embarrassingly open since the work of Gavoile et al. \cite{GPPR04}: the upper bound is $O(n^{1/2})$ bits per label (due to \cite{GawrychowskiU16} who shaved a log factor over \cite{GPPR04}), and the lower bound is $\Omega(n^{1/3})$.
The only known technique to prove polynomial lower bounds\footnote{The only result that somewhat deviate from this technique are $1.008\log n$ lower bound for nearest common ancestors in trees~\cite{alstrup2014near} and $1/8\log^2n$ lower bound for distance in trees~\cite{GPPR04,alstrup2015distance}. The gist of both of them is being able to argue about how much information can be shared by labels of two nodes. If the graph is not a tree, this seems very challenging.}
is to argue that labelling schemes are one way to compress graphs, and then use facts about the limits of graph compression.
For example, the lower bound for distance labelling of planar graphs \cite{GPPR04} follows because labels of size $O(n^{1/3-\eps})$ can be used to solve the metric compression problem using $O(k \cdot n^{1/3-\eps})$ bits, which contradicts the lower bounds above. 
In fact, the tight lower bound for metric compression of {\em weighted} planar graphs leads to a tight lower bound for labelling schemes \cite{GPPR04,AD16}.
Thus, to prove a tight lower bound of $\Omega(n^{1/2})$ for labelling schemes in unweighted planar graphs, the \emph{only} approach we have with current techniques is to negatively resolve Open Question~\ref{oq1}, e.g. by accomplishing Challenge~\ref{ch1}.

\paragraph{Routing and Dynamic Algorithms.} 
A compact routing scheme assigns names and tables to the nodes of a graph, so that each node $s$ can find out the first edge on the shortest path (or some approximate path) to any target node $t$ only using the name of $t$ and the local table stored at $s$.
There is a vast literature on the topic, seeking the best possible tradeoff between sizes of the tables and the stretch in many different graph families (we refer the reader to Peleg's book \cite{Peleg00book} and the extensive surveys \cite{Gav01survey,GP03survey}). 
For planar graphs, Abraham, Gavoille, and Malkhi \cite{AGM05} write: ``\emph{Surprisingly, for stretch 1, the complexity of the size of the routing tables is not known.}''
A simple upper bound is $\tilde{O}(n \cdot \sqrt{n})$ total table size, and an adaptation of the same Gavoille et al. construction gives a lower bound of $\Omega(n \cdot n^{1/3})$ \cite{AGM05}.
It is likely that accomplishing Challenge~\ref{ch1} would resolve this gap as well.
Yet another problem with similar state-of-the-art is the All Pairs Shortest Paths problem in \emph{dynamic} planar graphs.
Here, the goal is to have a data structure that supports efficient updates to the graph (edge additions or removals), and can answer shortest path queries efficiently.
The breakthrough algorithm of Fakcharoenphol and Rao \cite{FR06}, and the later optimizations \cite{Klein05,ItalianoNSW11,KaplanMNS12,GawK16}, achieve $\tilde{O}(n^{2/3})$ time for updates and queries.
The only framework for showing polynomial lower bound was recently proposed by Abboud and Dahlgaard \cite{AD16} who proved a lower bound of $n^{1/3-o(1)}$ under the popular APSP Conjecture \cite{RZ11,VassW10,AV14,AGV15,AVY15,Saha15,Dahlgaard16}.
Using their framework, accomplishing Challenge~\ref{ch1} directly leads to a higher lower bound of $n^{1/2-o(1)}$, as is known in the weighted case.

\medskip

History suggests that weighted planar graph metrics might be harder to work with, but they are never \emph{truly} harder.
In so many cases, a new algorithm for the unweighted case is followed by an almost-as-good algorithm for the weighted case, a few years later.
For example, a PTAS for the Travelling Salesman Problem in the unweighted planar metric was found in 1995 \cite{GKP95}, and then for the weighted case in 1998 \cite{AGK+98}.
Perhaps it is only a matter of time until our lower bounds for the unweighted metric match the weighted.

\subsection{Our Results}
Our first result is a new compression scheme for the planar graph metric, which achieves the information theoretically best possible bit complexity, up to log-factors. 
We give a \emph{positive} resolution to Open Question~\ref{oq1}, deem Challenge~\ref{ch1} to be infeasible, and show that unweighted planar graphs are inherently less complex than weighted ones; in fact, they admit a polynomially more efficient metric compression.

\begin{theorem}
Given an unweighted undirected planar graph on $n$ nodes and a subset $S$ of $k$ nodes, we can return a binary encoding of length $\tilde{O}(\sqrt{k\cdot n})$ from which all pairwise distances in $S$ can be recovered exactly. 
\end{theorem}

This shrinks the gap in our understanding of the planar metric compression problem from polynomial to polylogarithmic (removing this polylogarithmic gap remains an open question).
For comparison, when $k = \sqrt{n}$, we show that $\tilde{\Theta}(n^{3/4})$ bits are necessary and sufficient, while in the weighted case the bound is $\tilde{\Theta}(n)$.
Our encoding breaks the lower bound of Krauthgamer et al. \cite{KNZ14} for compressions using minors, and raises the question whether it can be matched via other forms of sparsification or \emph{graphical} compressions.

It is unclear whether our new compression scheme will lead to improved upper bounds for labeling, routing, or dynamic algorithms. In Section~\ref{sec:unit}, we discuss the difficulty in turning it into a labelling scheme. 
Still, it certainly shakes our beliefs about the right bounds for those problems. 
Even if better upper bounds are not possible, it is no longer a mere puzzle as in Challenge~\ref{ch1} that is standing in the way of higher lower bounds -- substantially new techniques and frameworks must be developed. 

\paragraph{Distance Oracles.}
Our first result was a mathematical advance in the understanding of the planar graph metric. 
Next, we use it algorithmically to achieve a new \emph{Subset Distance Oracle} that could be an appealing choice in many applications. 

A distance oracle is an encoding of a graph from which a pairwise distance can be queried efficiently.
Since the seminal paper of Thorup and Zwick \cite{TZ05}, a central subject of study in Graph Algorithms has been to understand the inherent tradeoff between the parameters of these distance oracles (see the survey by Sommer \cite{Sommer14survey}): The \emph{size} of the compression, the \emph{query time} for returning a distance, the \emph{error} in the answers, the \emph{preprocessing time} to construct the compression, and so on.

Many \emph{exact} distance oracles for planar graphs have been proposed \cite{Dji96,A+96,Cabello12,FR06,WN10,Nuss11}, mostly focusing on the tradeoff between space and query time, and giving $O(S)$ space and $O(n^2 / S)$ query time \cite{MS12}. 
Cohen-Addad, Dahlgaard, and Wulff-Nilsen \cite{CDW17} show that the technique of abstract Voronoi diagrams recently introduced into the field of planar graphs by Cabello \cite{Cabello17} leads to an oracle with $O(n^{5/3})$ space and $\tilde{O}(1)$ query time, suggesting that a better tradeoff is possible.

To get even better tradeoffs we might allow $(1+\eps)$ error \cite{Thorup04,Klein05,KKS11}: we can achieve very small $(1+\delta) n$ space and fast $\tilde{O}(1)$ query time.
Note that $o(n)$ space is impossible in this setting, no matter what query time we allow.
However, another natural way to get better tradeoffs is to restrict our attention to a subset of the nodes.
A \emph{Subset Distance Oracle} is a small space data structure that can efficiently return the distance between any pair of nodes from a set $S$ of $k$ nodes.
Here, for any $k=o(n)$, e.g. $k =\sqrt{n}$, our new compression scheme suggests that a distance oracle might have $o(n)$ space.

Subset distance oracles arise naturally.
In typical applications of distance oracles, one can predict that all queries will be among a subset of $k=o(n)$ nodes.
Space efficiency is often a high priority.
For example, if our graph is the national road network, one might be interested in a mobile app that can return the distance between any pair of bus stops.

Our second result is the first subset distance oracle with non-trivial space bounds. 
Notably, all previous distance oracles in the literature work equally well for weighted and unweighted graphs, while ours uses new techniques that are provably impossible for weighted graphs.

\begin{theorem}
Given an unweighted undirected planar graph on $n$ nodes and a subset $S$ of $k$ nodes, there is a polynomial-time algorithm that returns a data structure $X$ of size $\tilde{O}(\sqrt{k\cdot n})$ such that given $X$ and any pair of nodes in $S$ we can return the exact distance in $\tilde{O}(\min\{n^{3/4},\sqrt{k\cdot n}\})$ time.
\end{theorem}

The main open question left by our work is whether our query time can be improved, perhaps all the way down to $\tilde{O}(1)$.
This would be an essentially optimal distance oracle.
But even as it is, our query time is sublinear, and our space is sublinear for any $k=o(n)$, making it an appealing choice in applications with strict space constraints. 

\medskip
Finally, an intriguing and wide open question is to extend (any of) our upper bounds to \emph{directed} planar graphs.  Can we accomplish Challenge~\ref{ch1} if we allow directed edges?
Our tools heavily rely on the graph being undirected, yet it remains unclear if a higher lower bound can be proven for directed graphs. 

\subsection{Technical Overview}

We exhibit the first use of the  \emph{Unit-Monge} property to the algorithmic study of planar graphs.
It is well known that distances in a planar graph enjoy this property, due to the non-crossing nature of shortest paths in the plane, but prior to our work, only the \emph{(non-unit) Monge} property, was known to be algorithmically exploitable for planar graphs.
For the past few decades, it has been heavily utilized in numerous algorithms for problems related to shortest paths or minimum cuts in planar graphs (e.g.~\cite{FR06,Cabello12,Cabello17,MS12,KaplanMNS12,ItalianoNSW11,SubmatrixMonge,TALG10SP,MSMS,SPNegLengths2,GomoryHu,NussbaumOracle,LS11}), and beyond, in dozens of papers on computational geometry (e.g. \cite{KaplanMNS12,SubmatrixMonge,Kaplan1,Kaplan2,Kaplan3,AggarwalKlaweMoranShorWilber1987,SMAWK,KK89}) and pattern matching (e.g. \cite{Schmidt1998,CrochemoreLandauZiv,Algorithmica13,TiskinSODA}).
Meanwhile, the stronger \emph{Unit-Monge} property has only been exploited for algorithms on sequences where it has already led to several breakthroughs. We refer the reader to the 159-page monograph of Tiskin \cite{Tiskin} for an exposition of these applications.

Recall that we want to encode the distances among $k$ nodes in a planar graph.
Let us assume that we are lucky and all the nodes lie on a single face of the graph. 
Denote the nodes appearing on the face in order $s_1,\ldots,s_{k/2},t_1,\ldots,t_{k/2}$, and for simplicity assume that we only want to encode $s_i$-to-$t_j$ distances. Let $M$ be the $k/2 \times k/2$ matrix of distances so that $M[i,j]=d(s_i,t_j)$.
This matrix has the Monge property: For any $i,j$ we have that 
$M[i+1,j] - M[i,j] \le M[i+1,j+1] - M[i,j+1]$. This is because the $s_i$-to-$t_j$ shortest path and the $s_{i+1}$-to-$t_{j+1}$ shortest path must cross.
Moreover, it is \emph{Unit-Monge}, that is, $M[i+1,j]-M[i,j] \in \{-1,0,1\}$. This is because there is an edge between $s_i$ and $s_{i+1}$ and so the distances involving these nodes are always at most $1$ apart.

Our gains come from the fact that Unit-Monge matrices are compressible into $O(k \log k)$ bits!
For non-unit Monge matrices, the construction of Gavoille et al. implies an $\Omega(k^2)$ lower bound.
Another striking example for the extra power of the Unit-Monge property is the fact that (a compact representation of) the distance product of two $n \times n$ such matrices can be computed in $O(n \log{n})$ time \cite{TiskinSODA}, while for non-unit Monge matrices only $O(n^{2})$ algorithms are possible. 

The main issue for us, and in general when exploiting Monge properties, is that the nodes we care about do not necessarily lie on a cycle.
The simple solution is to \emph{add} a cycle connecting our $k$ nodes and assign weight $+\infty$ to the new edges so that they do not change the distances, or more formally, to \emph{triangulate} the graph.
After we do this, we have the Monge property, but because of the infinite weight edges, we do not have the unit-Monge property.
This solution is common to all the algorithms cited above that use the Monge property, and is quite reasonable when the graph is weighted to begin with.
For unweighted graphs, on the other hand, our work proves that it is too lossy and a more involved solution  
leads to much better results.

At a very high-level, our approach is to use a Baker-like decomposition into slices (vertices at consecutive levels of some specific BFS tree) whose boundaries are cycles, and to store distances to the slice boundaries. 
Observe that when we argued above that the unit Monge property holds because there is an edge between $s_i$ and $s_{i+1}$, we did not require that there is also an edge between $t_j$ and $t_{j+1}$. In our solution there is an edge between consecutive vertices on the boundary cycle of each slice. Therefore, even if we triangulate each slice using infinite weight edges, we can still exploit the unit Monge property when storing distances between certain vertices in a slice and the slice boundary.

The decomposition into slices is such that, after triangulation, the slices have small cycle separators. We recursively separate the vertices of the set $S$ within each slice  using small cycle separators. We store distances between separators and the slice boundary (using the unit Monge property) and between vertices of $S$ and separators (using the fact that separators are small). Significant technical issues arise with the nesting structure of slices. This gives rise to so-called {\em holes} in a slice. Dealing with multiple holes requires a detailed study of additional structural properties, and a more complicated recursive solution based on these properties. In essence, we show that whenever a na\"ive solution does not work in the presence of multiple holes, there is one hole that can be handled efficiently using a different approach. 

We believe it is very likely that other problems in unweighted planar graphs can be solved by exploiting the Unit-Monge property. Our near-optimal metric compression serves as a proof of concept that this is possible. However, technical challenges might have to be overcome in each specific application. In particular, the fast distance product algorithm for unit-Monge matrices~\cite{TiskinSODA} appears to be a strong and relevant technique that we are so far unable to exploit for solving problems in planar graphs.

\section{Preliminaries}

We assume basic familiarity with planar graphs and planar graph duality. We denote the primal graph by $G$ and the dual graph by $G^*$. For a spanning tree $T$ of $G$, we use $T^*$ to denote the spanning tree of $G^*$. It is well known~\cite{vonStaudt} that the set of edges of $G$ not in $T$ form a spanning tree $T^*$ of $G^*$. We often refer to $T^*$ as the {\em cotree} of $T$~\cite{Eppstein03}. 
For a spanning tree $T$ of $G$ and an edge $e$ of $G$ not in~$T$, 
the {\em fundamental cycle} of $e$ with respect to $T$ in $G$ is 
the simple cycle consisting of $e$ and 
the unique simple path in $T$ between the endpoints of $e$.

Given an assignment of nonnegative weights to the faces of $G$, 
we say that a simple cycle $C$ is a {\em balanced separator}
if the total weight of faces strictly enclosed by~$C$ 
and the total weight of faces not enclosed by~$C$ are each 
at most $5/6$ of the total weight.\footnote{It is more usual to require that the total weight is at most either
$2/3$ or $3/4$. However, in our particular application $5/6$ turns out to be more convenient.}
We often assign weights to vertices rather than to faces. Finding a balanced separator with respect to vertex weights reduces to the case of face weights (for each vertex, simply remove its weight and add it to an incident face).
It is well known (see, e.g.,~\cite{KleinMS13})  that in 
triangulated planar graphs there exists a balanced separator that is a fundamental cycle assuming
that no face has more than $1/2$ of the total weight
(in fact, this is true for any planar graph such that $T^*$ has maximum degree~3). 
For vertex-weights, if no vertex has more than 1/2 of the total weight and the graph is triangulated and there are no self loops then by evenly transferring the weights to faces we obtain that no face receives more than 1/2 of the total weight (because every node is incident to at least two faces) and we can invoke the face-weights version of the balanced separator.
Many planar graph algorithms triangulate the graph by adding edges to ensure that short balanced cycle separators exists. The lengths of the added edges is set to be sufficiently large so as not to change distances in the graph. This is clearly not possible in unweighted planar graphs, and is one of the obstacles we will need to overcome.

\subsection{The Monge and Unit-Monge properties}

One of the main tools we use for succinct representation of distances in
unweighted undirected planar graphs is the {\em unit Monge property}. We next
prove a sequence of lemmas that utilize this property to efficiently store
distances between vertices on cycles. We begin with encoding distances between
disjoint sets of vertices on a single face (Lemma~\ref{lem:unit}), then
encoding all-pair distances between the vertices on a single face
(Lemma~\ref{lem:unit2}), then encoding all-pair distances between
the vertices on a single simple  cycle (Lemma~\ref{lem:cycle}), and finally,
encoding the distances between the vertices of two faces
(Lemma~\ref{lem:hole}).

\begin{lemma}
\label{lem:unit}
Let $C=(v_1,v_2,\ldots,v_{|C|})$ be the cyclic walk of a face of a planar graph
partitioned into two parts
$C_1=(v_1,v_2,\ldots,v_{\ell})$ and $C_2=(v_{\ell+1},v_{\ell+2},\ldots,v_{|C|})$.
Then,
for any subset $C'_2$ of $C_2$, all distances between vertices of $C_1$ and vertices of
$C'_2$ can be encoded
in $\tilde\Oh(|C_1|+|C'_2|)$ bits.
\end{lemma}

\begin{proof}
Let $C'_2=\{v_{p_1},v_{p_2},\ldots,v_{p_s}\}$. We define an $\ell \times s$ matrix $M$ such that
$M[i,j]$ equals the distance in $G$ between $v_i$ and $v_{p_j}$. The matrix $M$ is Monge, that is
\[
M[i+1,j] - M[i,j] \le M[i+1,j+1] - M[i,j+1]  
\]
for any $i\in [1,\ell-1]$ and $j\in [1,s-1]$. This is because the shortest $v_i$-to-$
v_{p_j}$ and $v_{i+1}$-to-$v_{p_{j+1}}$ paths must necessarily cross. Furthermore, the matrix $M$ is unit-Monge, that is

\[
M[i+1,j]-M[i,j] \in \{-1,0,1\}
\] 
for any $i\in [1,\ell-1]$ and $j\in [1,s]$, because there is an edge
$(v_i,v_{i+1})$.
Consequently, for any $i\in [1,\ell-1]$, the sequence of differences
$M[i+1,j]-M[i,j]$
is nondecreasing and contains only values from $\{-1,0,1\}$, so can be encoded
by storing the positions of the first 0 and the first 1. Storing these positions
for every
$i \in [1,\ell-1]$ takes $\tilde\Oh(\ell)$ bits. To encode the whole $M$, we
additionally store $M[0,j]$ for every $j\in [1,s]$ using $\tilde\Oh(s)$ bits.
\end{proof}

\begin{lemma}
\label{lem:unit2}
Let $C=(v_1,v_2,\ldots,v_{|C|})$ be the cyclic walk of a face of a planar graph. Then,
all distances between vertices of $C$ can be encoded in $\tilde\Oh(|C|)$
bits.
\end{lemma}

\begin{proof}
We recursively encode all distances between vertices from a contiguous fragment
of
$C$ using Lemma~\ref{lem:unit}. We start with the whole $v_1,v_2,\ldots,v_{|C|}$.
To encode the distances between all vertices $v_i,v_{i+1},\ldots,v_{j}$,
where
$i<j$, we set $m=\lfloor (i+j)/2\rfloor$ and proceed as follows:
\begin{enumerate}
\item Recursively encode the distances between all vertices 
$v_i,v_{i+1},\ldots,v_m$.
\item Recursively encode the distances between all vertices 
$v_{m+1},v_{m+2},\ldots,v_{|C|}$.
\item Apply Lemma~\ref{lem:unit} with $C_1=(v_i,v_{i+1},\ldots,v_m)$ and
$C'_2=\{ v_{m+1},v_{m+2},\ldots,v_{j}\}$.
\end{enumerate}
The total size of the encoding is described by the recurrence 
$T(s)=\tilde\Oh(s)+2T(s/2)$, hence solves to $\tilde\Oh(|C|)$.
\end{proof}

\begin{lemma}
\label{lem:cycle}
Let $C=(v_1,v_2,\ldots,v_{|C|})$ be a simple cycle in a planar graph. Then, all
distances between
vertices of $C$ can be encoded in $\tilde\Oh(|C|)$ bits.
\end{lemma}

\begin{proof}
Consider the two planar graphs $G_{out}$ ($G_{in}$) obtained by removing all
vertices enclosed (not enclosed) by $C$. $C$ is the cyclic walk of a face in
$G_{in}$ and $G_{out}$,
hence we can apply Lemma~\ref{lem:unit2} to store the distances in $G_{out}$ and in $G_{in}$ between
vertices of $C$. This is enough to encode the distances in $G$ between vertices of $C$,
 as any such shortest path can be partitioned into shortest paths between two
vertices
of $C$ such that each of these paths exists in either $G_{in}$ or $G_{out}$.
\end{proof}

\begin{lemma}
\label{lem:hole}
Let $C_{ext}=(v_1,v_2,\ldots,v_{|C_{ext}|})$ and $C_{int}=(u_1,u_2,\ldots,v_{|C_{int}|})$ be the cyclic walk of two faces of a planar graph. Then, all
distances between a prefix $C'_{ext}$ of $C_{ext}$ and  any subset $C'_{int}$ of $C_{int}$ can
be encoded in $\tilde\Oh(|C'_{ext}|+|C'_{int}|)$ bits.
\end{lemma}

\begin{proof}
We first
choose a shortest path $P$ between $C'_{ext}$ and $C_{int}$ and let $v_i$ and
$u_j$
be its endpoints. We make an incision along $P$ and apply Lemma~\ref{lem:unit}
to
encode the distances between $C'_{ext}$ and $C'_{int}$ corresponding to
shortest
paths that do not cross $P$ using $\tilde\Oh(|C'_{ext}|+|C'_{int}|)$ bits of
space. It remains to encode
distances corresponding to shortest paths that do cross $P$. Without loss of generality $P$ connects $v_1$ and $u_1$. We
orient
$C_{ext}$ and $C_{int}$ so that after making an incision along $P$ the vertices $v_2$ and $u_2$ are adjacent to the endpoints of $P$. 

\begin{figure}[h]
\begin{center}
\includegraphics[width=0.4\textwidth]{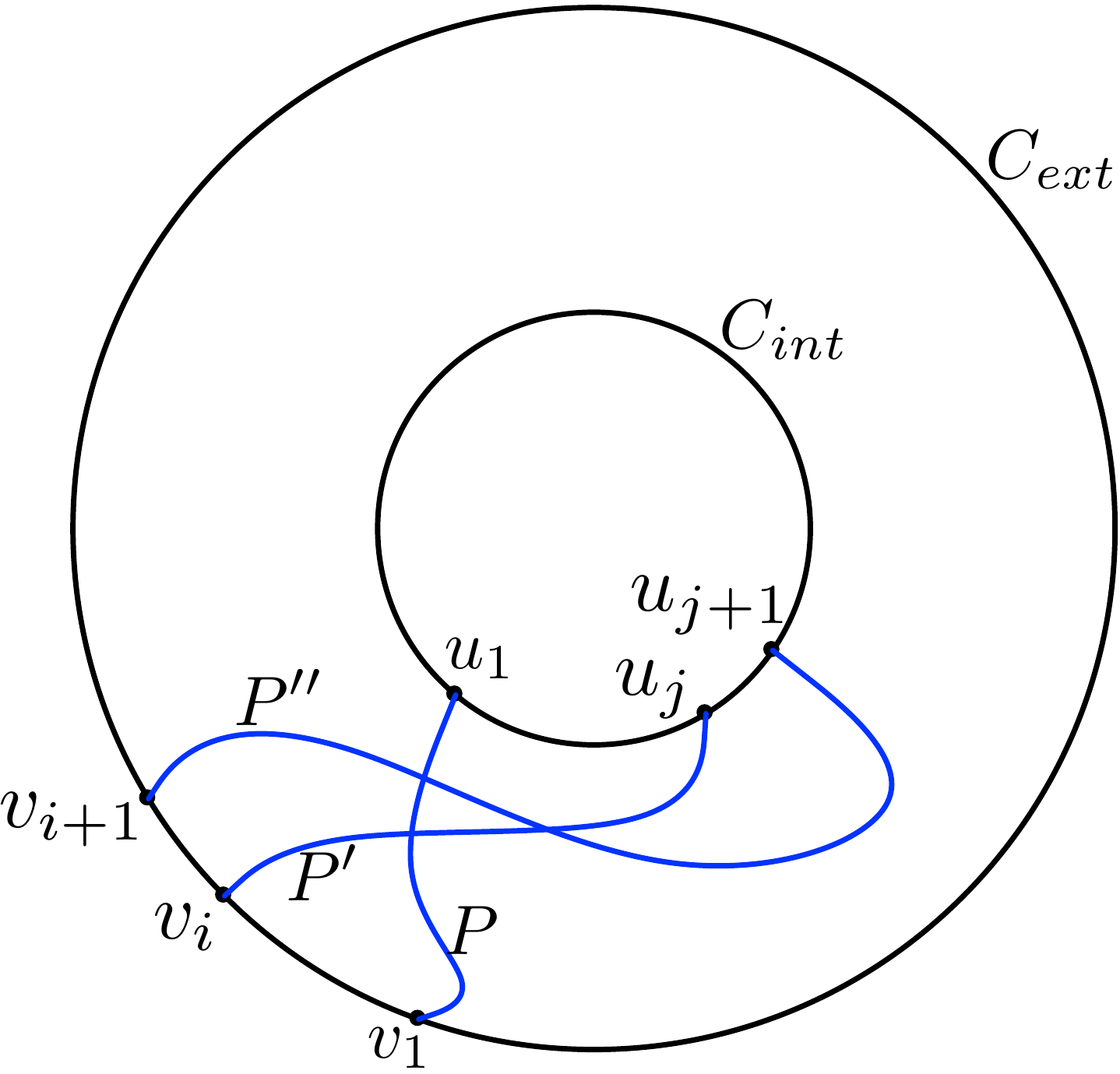}
\end{center}
\caption{The Monge property in Lemma~\ref{lem:hole}.\label{fig:cross}}
\end{figure}

Consider a shortest path $P'$ from $v_i$ to $u_j$ crossing $P$, see
Figure~\ref{fig:cross}. Because
both $P$ and $P'$ are shortest paths, $P'$ can be assumed to cross $P$ exactly once. Similarly, a shortest path $P''$
from
$v_{i+1}$ to $u_{j+1}$ crossing $P$ can be assumed to cross $P$ exactly once. 
We claim that $P'$ must cross $P''$. Otherwise, by considering an incision along $P'$ we can conclude that $P''$ crosses $P$ an even number of times but this is a contradiction. Therefore, any such $P'$ and $P''$ must cross. This means that the matrix $M$,
where $M[i,j]$ is set to be the length of a shortest path between $v_i$
and $u_j$ crossing $P$ once, is Monge. That is, 
\[ M[i+1,j+1] + M[i,j] \geq M[i+1,j]+M[i,j+1]. \]
Additionally, $M[i,j]-M[i+1,j] \in \{-1,0,1\}$ because
$(v_i,v_{i+1})$
is an edge. We can hence apply the reasoning from Lemma~\ref{lem:unit} to
encode $M$ using $\tilde\Oh(|C'_{ext}|+|C'_{int}|)$ bits.
\end{proof}

\section{The Encoding}
\label{sec:upper}

Our encoding is based on decomposing the input graph $G$ into
\emph{slices}. To define the slices, recall
the {\em face-vertex incidence graph} $FV(G)$ of a planar graph $G$: It has a vertex for every vertex $v$ of $G$ and a vertex for every face $f$ of $G$, and if a vertex $v$ of $G$ is incident to a face $f$ of $G$ then there is an edge between their corresponding vertices in $FV(G)$.

We run a breadth-first search in $FV(G)$, starting from
the node representing the infinite face of $G$.
After every even number of steps, the yet unexplored part of the graph
can be decomposed into a number of connected components, the boundary of each
being
a simple cycle.
More formally, we assume that the infinite face of $G$ is a triangle by
enclosing the whole graph in a triangle, which is connected to one of the
original vertices with a single edge. 
We define the level of a face $f$ or a vertex $v$ of $G$ to be its
depth in the BFS tree of $FV(G)$. Thus, e.g., the level of the infinite
face of $G$ is zero, and the level of the vertices incident to the infinite
face of $G$ is 1. For each even integer $i \geq 2$, we define
$\mathcal K_{\geq i}$ to be the set of connected components of the
 subgraph of $G$
induced by the faces with level at least $i$. We call each component
$K \in \mathcal K_{\geq i}$ a level-$i$ component.
We use a tree $\mathcal K$ called the {\em component tree} of $G$ to capture the nesting of level
components.  
The nodes of $\mathcal K$ are the level components of
$G$. A level component $K$ is an ancestor of a level component $K'$ in
$\mathcal K$ if the set of faces in $K$ contains the set of faces in
$K'$.
Since we assume that the infinite face of $G$ is a simple triangle,  $\mathcal
K$ is indeed a tree whose root is the component corresponding to the set of all
faces of $G$ except
for the infinite face. 

The boundary of a component $K$ is the set of edges that are incident
to a face in $K$ and to a face not in $K$. It is not difficult to see
that the boundary of each component $K$ is a simple cycle in $G$, and that the boundaries of different components are edge-disjoint. See~\cite{KleinMS13, PlanarBook} for these and other properties of components and the component tree.
For a node $k \in \mathcal K$, we associate $k$ with the boundary
cycle $C_k$ of the level component
represented by $k$, and define the cost of $k$ denoted $\cost(k)$ to be
the length of $C_k$.
For
example, for the root $r$ of $\mathcal K$ we have that $C_r$ is a
triangle and that $\cost(r)=3$.

\begin{lemma}
\label{lem:slices}
For any $w\geq 1$, there exists  $\delta\in [0,w)$ such that the total cost of
all nodes
of $\mathcal K$ at depth $\delta,\delta+w,\delta+2w,\ldots$ is $\Oh(n/w)$.
\end{lemma}

\begin{proof}
$\sum_{v\in\mathcal{K}} \cost(v)=\Oh(n)$ because cycles corresponding to the
nodes of $\mathcal{K}$ are  pairwise edge disjoint. Let $S(\delta)$ consist
of all nodes of $\mathcal{K}$ at depth $\delta,\delta+w,\delta+2w,\ldots$.
Then $S(\delta)\cap S(\delta') = \emptyset$ for $\delta\neq\delta'$ and
$\sum_{\delta\in[0,w)} \sum_{v\in S(\delta)} \cost(v)=\Oh(n)$,
so there exists $\delta\in [0,w)$ such that $\sum_{v\in S(\delta)}\cost(v) =
\Oh(n/w)$
as claimed.
\end{proof}

To define the slices we apply Lemma~\ref{lem:slices} and call the nodes of
$\mathcal{K}$ at
depth $\delta,\delta+w,\delta+2w,\ldots$ \emph{marked}. The root of
$\mathcal{K}$
is also marked. Then, for every marked node $v\in \mathcal{K}$, 
the slice of $v$ is the subgraph of $G$ enclosed by $C_v$ and not
strictly enclosed by $C_u$ for any marked descendent $u$ of $v$.
The embedding of slices is inherited from the embedding of $G$. 
Thus, the boundary of the infinite face of the slice $s$ of $v$ is $C_v$.
The cycle $C_v$ is also
called \emph{the boundary of the slice $s$}. 
Each cycle $C_u$ corresponding to a marked descendant $u$ of $v$ such that
there are no other marked nodes on the $v$-to-$u$ path becomes a face
in the slice $s$. Such a face is called a \emph{hole} of $s$,
and $C_u$ is called the {\em boundary of the hole}. Note that, by definition, $C_u$ is the boundary of the level component that is embedded in the hole $u$.
 Because the total
cost of all marked nodes is $\Oh(n/w)$ and the cost of the root is 3, the total
size of all boundaries in all slices is $\Oh(n/w)$. Additionally, by
construction, for any slice $s$, 
a breadth-first search of $FV(s)$, the face-vertex incidence graph of $s$,
starting at the infinite face of $s$, terminates after
$\Oh(w)$ iterations and every hole is a leaf in the obtained breadth-first
search tree.

By definition of slices, each slice contains faces and vertices at $\Oh(w)$
consecutive levels. 
We would like to use in our solution short (i.e., $\Oh(w)$) fundamental cycle
separators within each slice. 
However, the diameter of a slice is not necessarily $\Oh(w)$ because face sizes
may be large. To deal with this issue we triangulate the faces so that a BFS
tree of a slice will have depth $\Oh(w)$, and will be consistent with the BFS tree
of $FV(s)$. 

\begin{lemma}
\label{lem:triangulation}
We can triangulate all  
faces of a slice $s$ so that a BFS starting from the external
face produces a spanning tree $T_s$ with the property that vertex $v$ is the
parent of vertex $u$ in $T_s$ if and only if $v$ is the grandparent of $u$ in
the BFS tree of $FV(s)$.
\end{lemma}

\begin{proof}
Let $T_{FV}$ be the BFS tree of $FV(s)$. If $v$ and $u$ are incident to the same face in $s$, and $v$ is a grandparent of $u$, and $vu$ is not an edge in $s$, we add $vu$ as an artificial triangulation edge to $s$. Adding these edges can be done consistently with the embedding of $s$ because the path in $T_{FV}$ can be embedded on the same plane as $s$ such that $s$ and $T_{FV}$ only intersect at vertices of $s$. See Figure~\ref{fig:slices}.
We introduce an artificial vertex $v_s$ embedded in the infinite face of $s$ and
triangulate the infinite face of $s$ by adding edges between $v_s$ and every
vertex of the infinite face of $s$.  
Similarly, we triangulate each hole $h$ of $s$ by introducing an artificial vertex $v_h$, embedded in $h$, and adding edges between $v_h$ and every vertex on the boundary of $h$. 
Any remaining non-triangulated faces can be triangulated arbitrarily.
Since for every grandparent to grandchild path in $T_{FV}$ there is a corresponding edge in the triangulation of $s$, there exists a BFS tree $T_s$ rooted at the artificial vertex $v_s$ that satisfies the statement
of the lemma. Note that all the artificial vertices embedded in holes of $s$
are leaves of $T_s$, and hence satisfy the statement of the lemma
vacuously. 
\end{proof}

\begin{figure}
\centering
\begin{minipage}{0.5\textwidth}
  \centering \vspace{-0.4in}
  \includegraphics[width=.8\textwidth]{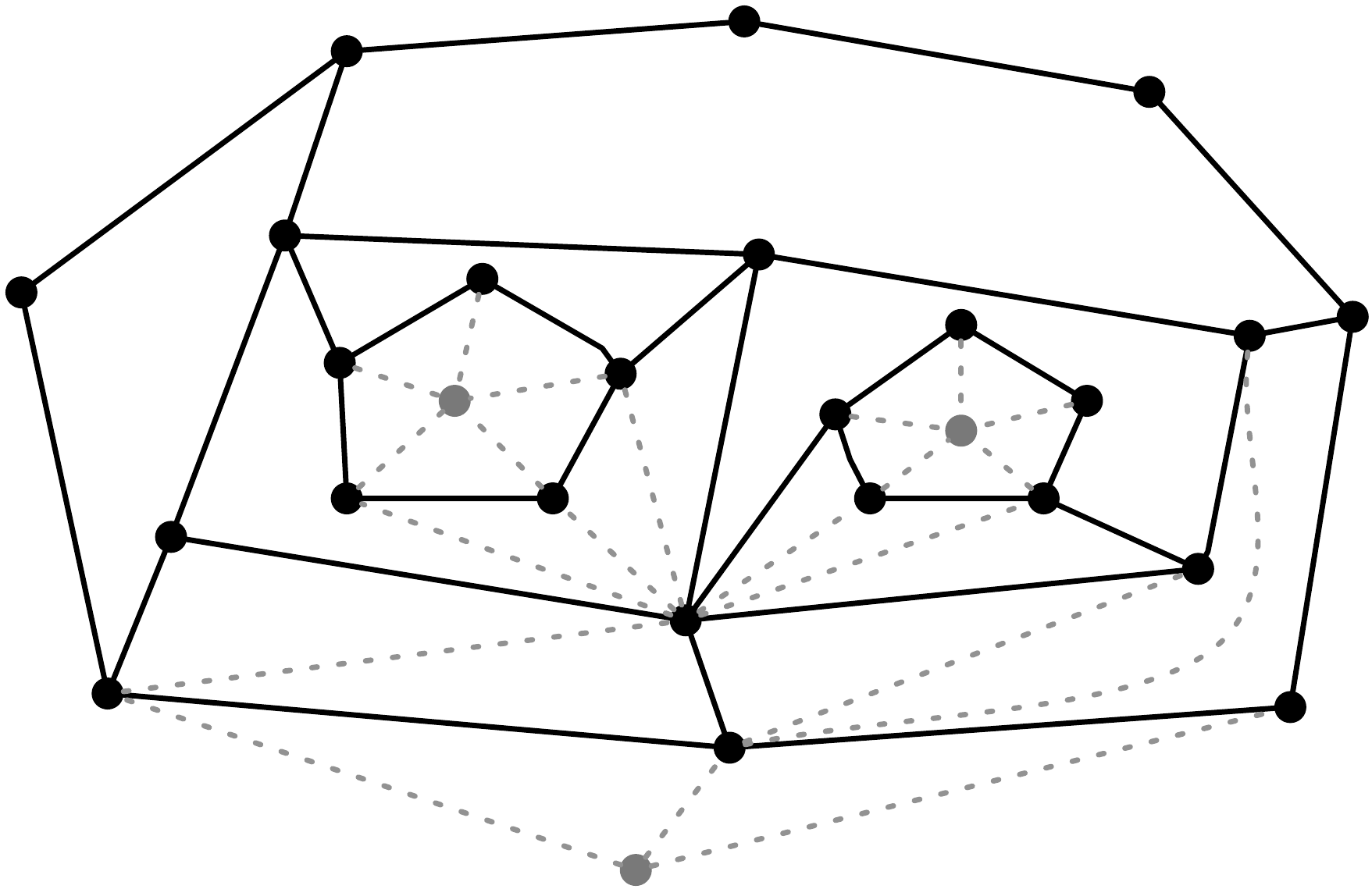}
\end{minipage}%
\begin{minipage}{0.5\textwidth}
  \centering
  \includegraphics[width=.8\textwidth]{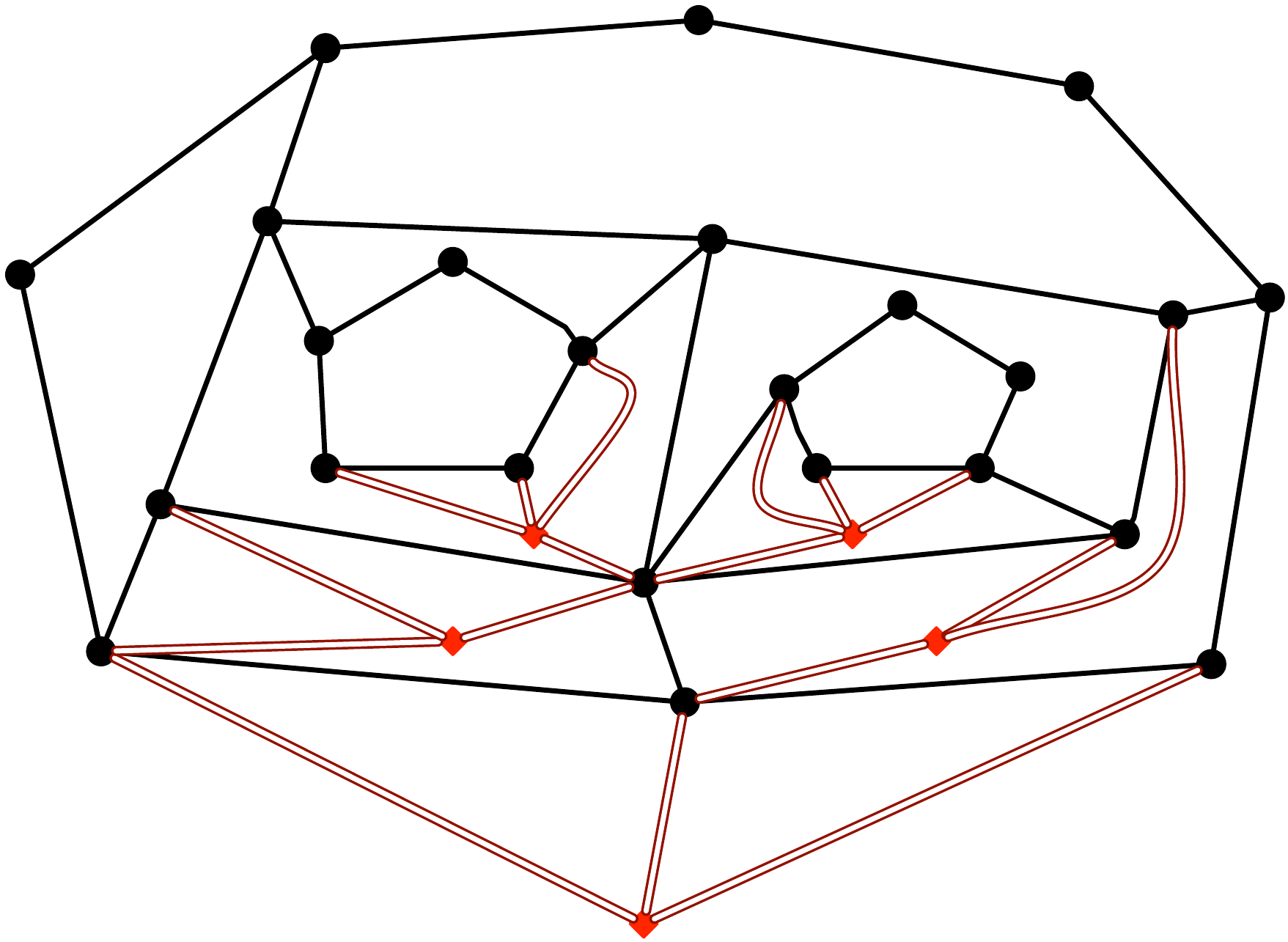}
\end{minipage}
\begin{minipage}{0.5\textwidth}
  \centering \vspace{0.05in}
\includegraphics[width=0.8\textwidth]{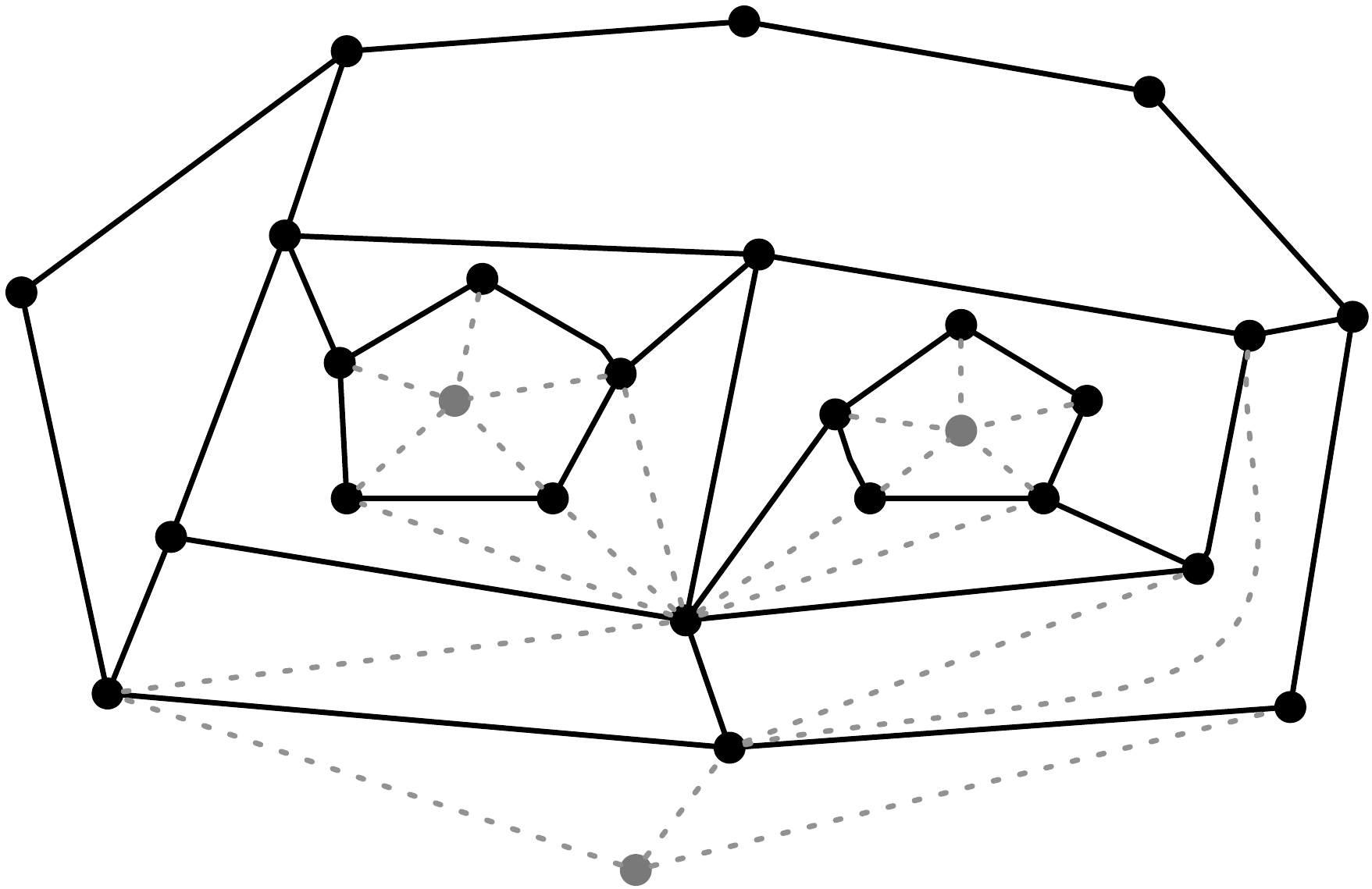}
\end{minipage}%
\begin{minipage}{0.5\textwidth}
  \centering \vspace{0.051in}
\includegraphics[width=0.8\textwidth]{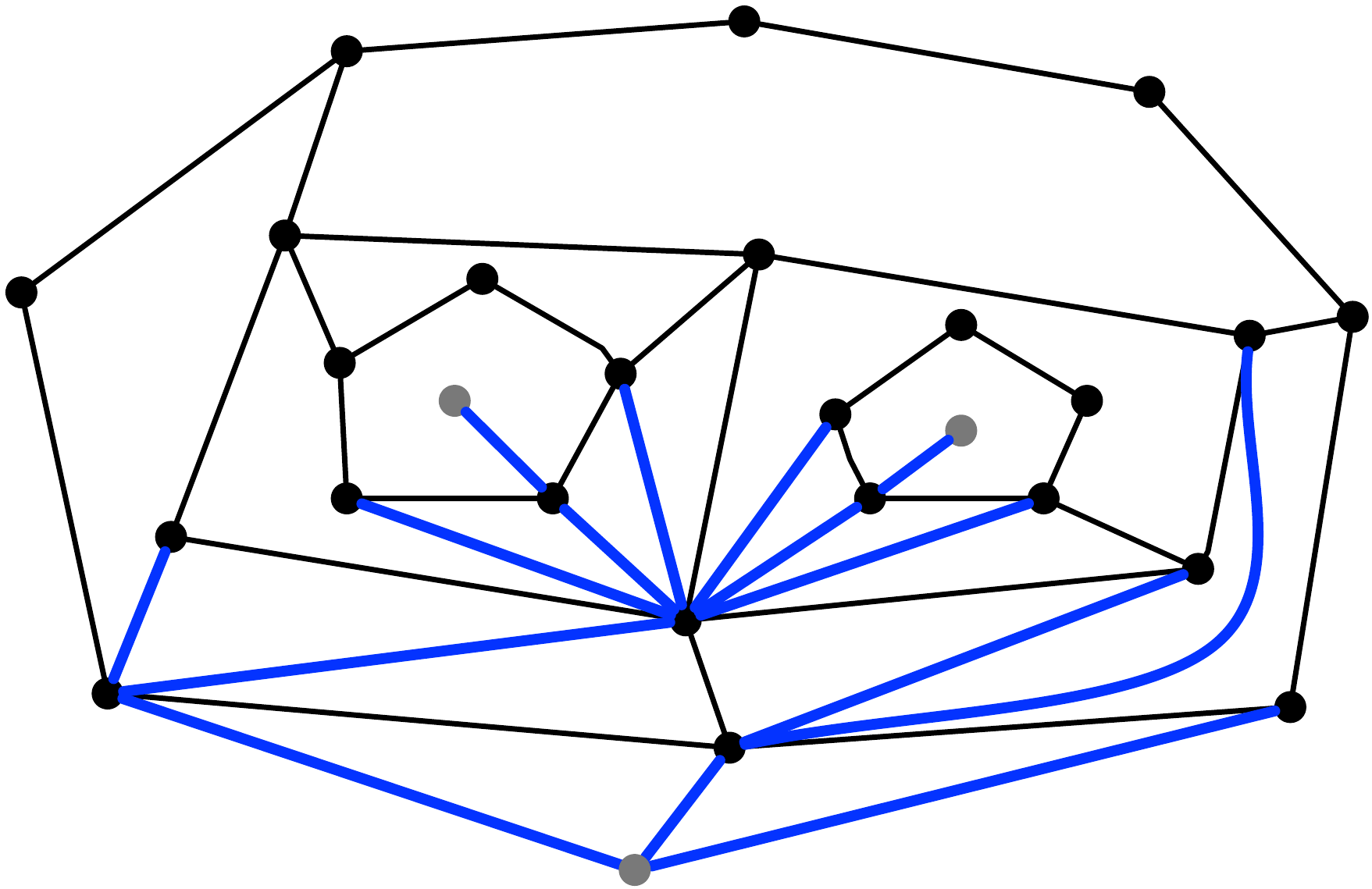}
\end{minipage}
\caption{Triangulating a slice. The vertices and edges of a slice $s$ with two holes are shown in solid black. In this example $w=2$. Edges of $T_{FV}$, the BFS tree of $FV(s$) are shown in double red lines. Only some of the edges of $T_{FV}$ are shown to avoid clutter. Artificial triangulation edges are shown in dashed gray. The BFS tree $T_s$ of the triangulation of $s$ is shown in blue. \label{fig:slices}}
\end{figure}

Let $s'$ be the graph obtained from the slice $s$ after applying Lemma~\ref{lem:triangulation}.
Let $T_s$ be the BFS tree of $s'$. Note that, any fundamental cycle $C$ w.r.t. \!$T_s$ consists of two paths in $T_s$, each consisting of $\Oh(w)$ vertices due to the triangulation. However, $C$ may use edges that are not original edges of $s$ (i.e., artificial triangulation edges). We do not want to consider such edges when dealing with distances, because distances in $s'$ differ from distances in $s$. To this end we use the notion of a Jordan curve. A Jordan curve in $s$ is an embedded curve that intersects the embedding of $s$ only
at vertices of $s$. Since the embedding of the triangulation $s'$ is consistent with that of $s$, each path in $T_s$ is a Jordan curve in $s$. We say that $T_s$ is a {\em Jordan tree} in $s$. In particular, any fundamental cycle w.r.t. \!$T_s$ is a Jordan cycle (closed Jordan curve) in $s$. 
We next describe how the tree $T_s$ can be used to recursively decompose $s$ into  subgraphs called {\em regions}.   

A region $R$ is a subgraph of $s$. The boundary of $R$ is defined as the set of vertices of $R$ that are incident (in $G$) to both an edge in $R$ and to an edge not in $R$. Thus, for example, the boundary of the region consisting of the entire slice $s$ consists of the external boundary $C_{ext}$ of $s$ and of the boundaries of all the holes of $s$. 
Let $R$ be a region. Let $C$ be a fundamental cycle w.r.t. \!$T_s$. The tree $T_s$ may contain edges that are not edges of $R$ (either because they are triangulation edges, or because they are edges of $s$ that do not belong to the region $R$). Since the embedding of $T_s$ is consistent with the embedding of any subgraph of $s$, $C$ is a Jordan cycle in $R$. The operation of separating $R$ using $C$ yields two subgraphs. One is the subgraph induced by the faces of $R$ strictly enclosed by $C$ and the other is the subgraph induced by the faces of $R$ not strictly enclosed by $C$. This view of $T_s$ as a Jordan tree in any region allows us to reuse the same tree $T_s$ throughout the recursive decomposition.

This recursive process can be described by a binary tree $\mathcal T_s$. Each node $v$ of
$\mathcal T_s$ corresponds to a region (subgraph) $R_v$ of $s$. The root of
$\mathcal T_s$ is the entire slice $s$. Each non-leaf node $v$ of $\mathcal T_s$ is associated with a (Jordan) fundamental cycle separator of $T_s$, which we denote $Sep_v$. The regions of the two children of $v$ are the regions obtained by separating $R_v$ with the Jordan cycle $Sep_v$.

\subsection{The simplified case of a single hole}\label{singleholesection}
We begin with the simplified case, in which we assume
that each slice has a single hole. This is the case, for example, when the
input planar graph is a grid (with possibly subdivided edges).

First we use Lemmas~\ref{lem:cycle} and~\ref{lem:hole} to store, for each slice
$s$ with external boundary $C_{ext}$, and a single hole $h$ with external boundary $C_h$ the following distances. The
{\bf boundary-to-boundary} distances: the distances (in $s$) among the vertices of $C_{ext}$, and the {\bf hole-to-boundary} distances: the
distances (in $s$) between the vertices of $C_{ext}$ and the vertices of $C_h$.
 
Boundary-to-boundary and hole-to-boundary distances encode distances ``between slices''.
We also need to encode distances ``within slices''. 
We will use the fact that $s$ has a spanning tree of depth $\Oh(w)$ to decompose $s$ into regions, each containing a single distinguished node (i.e., node of $S$), and having a boundary that consists of $\Oh(w)$ vertices. Then we can afford to store, for each distinguished node, its distance to the boundary of its region, and, using the unit-Monge property, to also store the distances between the $\Oh(w)$ vertices on the boundary of each region $R$ to the vertices of $C_{ext}$ and $C_h$ (i.e., the boundary of $s$) that belong to $R$. These distances will suffice for reconstructing the distance between any pair of distinguished nodes.
 
Let $S_s$ denote the set of distinguished vertices in slice $s$. We use
fundamental (Jordan) cycle separators w.r.t. \!the tree $T_s$ to recursively
divide $s$ into regions, until each region contains a single distinguished vertex. At each recursive step we separate a region $R$ into two subregions by choosing a
fundamental cycle separator w.r.t. \!$T_s$ that balances the number of
distinguished vertices in $R$ (i.e., assigning unit weight to each distinguished vertex in $R$ and
zero weight to all other vertices).  
Note that, since we use balanced separators, the depth of the recursion tree $\mathcal T_s$ is $\Oh(\log |S_s|) = \Oh(\log n)$. Recall that the fundamental cycle separators w.r.t. \!$T_s$ do not cross each other, and, by construction of
$T_s$ in Lemma~\ref{lem:triangulation}, each fundamental cycle separator crosses each of the external boundary of $s$ and the hole of $s$ at most twice. Therefore, the
boundary of each region $R_v$ corresponding to a node $v$ in the recursive decomposition tree $\mathcal T_s$ contains  $\Oh(\log n)$ vertex disjoint maximal subpaths of
$C_{ext}$, and $\Oh(\log n)$ vertex disjoint maximal subpaths of $C_h$.

At the
step of the recursive decomposition corresponding to node $v \in
\mathcal T_s$ with separator $Sep_v$ and two children $u,w$, we store {\bf $\boldmath S$-to-separator} distances: explicitly store the distances
(in $R_v$) between every vertex of $S$ in $R_v$ and every vertex of $Sep_v$, {\bf separator-to-boundary} distances and {\bf separator-to-hole} distances:
for $i \in \{u,w\}$, the distances (in $R_i$) between every vertex of $Sep_v$ and every maximal
subpath of $C_{ext}$ or $C_h$ on the boundary of $R_i$, using Lemma~\ref{lem:unit} or Lemma~\ref{lem:hole} (depending whether they lie on a single or two faces of $R_i$).  
Finally, for every leaf $v \in \mathcal T_s$, we store {\bf $S$-to-boundary} distances and {\bf $S$-to-hole} distances: the distance between the unique distinguished vertex in $R_v$ to
every vertex of $C_{ext}$ or $C_h$ on the boundary of $R_v$.

\paragraph{Analysis.}
We first show that the total space is $\tilde \Oh(\sqrt{k\cdot n})$, and then show
that the distances between any pair of vertices in $S$ can be recovered using just
the information we stored.
Since the total size of all slice boundaries is $\Oh(n/w)$, storing the boundary-to-boundary distances and the hole-to-boundary distances takes $\tilde \Oh(n/w)$ using Lemmas~\ref{lem:cycle}
and~\ref{lem:hole}. 
Since the depth of $\mathcal T_s$ is $\Oh(\log n)$, each vertex of $S_s$
belongs to $\Oh(\log n)$ regions in the decomposition of $s$.
Since, in addition, $|Sep_v| = \Oh(w)$ for every $v \in \mathcal T_s$, the total space required for storing the $S$-to-separator distances is $\tilde \Oh(k\cdot w)$.
Consider a region $R$ of a slice $s$. Recall that the vertices of $C_{ext}$
($C_h$) that belong to $R$ lie on $O(\log n)$ vertex disjoint maximal subpaths
of $C_{ext}$ ($C_h$). The endpoints of each such maximal subpath may belong to another region $R'$ at the same depth in $\mathcal T_s$. Therefore, $R$ shares $O(\log n)$ vertices of $C_{ext}$ ($C_h$) with other regions at the same depth in $\mathcal T_s$.
Finally, recall that the number of regions of $s$ is $\tilde \Oh(|S_s|)$.
Therefore, using Lemma~\ref{lem:unit} or Lemma~\ref{lem:hole}, the total space for storing the separator-to-boundary and separator-to-hole distances is $\tilde \Oh(n/w+k+k\cdot w)$.
In more detail, let $c_i$ be the total number of slice/hole boundary vertices in the $i$-th slice.
Then, in every slice every boundary/hole vertex that is not an endpoint of a maximal subpath
contributes at most once at each level of recursion. At each level, we have at most $k$ recursive calls,
so at most $\Oh(k\log n)$ maximal subpaths and at most $k$ fundamental cycle separators.
Therefore, the total space is $\Oh((\sum_{i} c_{i}+k\log n+k\cdot w)\log n)=\tilde\Oh(n/w+k+k\cdot w)$.
Storing $S$-to-boundary distances and $S$-to-hole distances at the leaves of the recursion tree requires total $\tilde \Oh(k + n/w)$ bits since each boundary or hole vertex belongs to exactly one leaf region, except for $O(k\log n)$ vertices (endpoints of maximal subpaths).
Choosing $w=\sqrt{n/k}$ proves the space bound.

Finally, we prove that the distances between any pair of vertices in $S$ can be recovered using just the information we stored.
For any $x,y \in S$, if a shortest $x$-to-$y$ path does not leave $s$ then $x,y \in S_s$, and the distance can be obtained using the
$S$-to-separator distances stored in the lowest common ancestor of the regions of $x$ and $y$ in $\mathcal T_s$. 
Otherwise, let $P$ be a shortest path between vertices $x \in S_{s'}$ and $y \in S_{s}$ (where $s'$ is either $s$ or, wlog, enclosed by the hole of $s$). 
Let $P[i,j]$ denote the subpath of $P$ between vertices $i$ and $j$. 
Let $v$ be the first vertex
of $P$ that belongs to the boundary of $s'$ or to the boundary of a hole of $s'$. 
If $P[x,v]$ contains some vertex of a fundamental cycle separator used in processing $s'$, let $u$ be
the last vertex of $P$ that belongs to the earliest such separator.
If $u$ does not exist, then the length of $P[x,v]$ is stored as an $S$-to-boundary or an $S$-to-hole distance. 
If $u$ exists then the length of $P[x,u]$ is stored as a $S$-to-separator distance,
and the length of $P[u,v]$ is stored as a separator-to-boundary or separator-to-hole distance. Let
$w$ be the last vertex of $P$ that belongs to the boundary of $s$. The length
of $P[v,w]$ can be computed from boundary-to-boundary and hole-to-boundary distances since $P[v,w]$
can be decomposed into subpaths between boundary vertices of slices. The length
of the suffix $P[w,y]$ can be computed in a similar manner to that of the
prefix $P[x,v]$. 

\subsection{The general case}

A difficulty that arises in the presence of multiple holes is that since the number of holes is not bounded by a constant, we cannot afford to store distances involving holes. For example, storing hole-to-boundary  distances between the external boundary $C_{ext}$ of a slice $s$ and the boundary of each hole of $s$ requires $\Omega(|C_{ext}|)=\Omega(n/w)$ bits per hole. Since the number of holes can be $\Omega(n)$, the total space could be $\Omega(n^2/w)$.

The role of storing distances involving boundaries of holes was to allow the recovery of distances to distinguished vertices  enclosed in these holes. We modify our approach for processing a slice $s$ to take into account the distinguished vertices enclosed in holes of $s$ as well as the distinguished vertices in $s$ itself. As in the single hole case, the slice $s$ will be recursively divided using fundamental cycle separators.
For any region $R$ encountered along the recursive process, let $S^R$ denote the subset of the distinguished vertices in $R$, as well as those enclosed by any hole in $R$.  Thus, for example, $S^s$ is the set of all vertices in $S$ that are enclosed (in $G$) by the external boundary of slice $s$.
We say that a Jordan cycle separator $C$ of a region $R$ is \emph{good} if it is balanced w.r.t. \!$S^R$ and does not go through any hole of $R$. The problem with Jordan separators that go through some hole $h$ is that they partition the distinguished vertices enclosed by $h$ in an unspecified way since these distinguished vertices are not represented in $R$. 
It is not hard to see that if a good separator always exists then we do not need to store any distances involving holes. 

In reality we cannot always find a good separator. 
Consider, for example, the case where some hole $h$ of a region $R$ encloses most of the vertices of $S^R$. Clearly, a separator that is balanced w.r.t. \!$S^R$ must go through $h$. Thus, there is no good separator in such a case.
We show, however, that we can always either find a good separator, or there exists some hole (which we call a \emph{disposable} hole) that can be dealt with in a special way.
This is reminiscent of recursive procedures based on heavy path decomposition, where heavy nodes (disposable holes in our case) are treated differently than light ones.  
We guarantee that, in either case, each resulting subregion contains only a constant fraction of $S^R$, so the depth of the recursion is $O(\log n)$. 
We next explain the details.

\paragraph{Good separators and disposable holes.}
Let $R$ be a region. We define the weight of each vertex $v$ of $R$ to be $1$ if $v$ is a distinguished vertex. For each hole $h$ of $R$, we define the weight of the artificial vertex $v_h$ embedded in $h$ to be the  number of distinguished vertices
strictly enclosed (in the whole graph $G$) by the boundary of $h$. All other vertices are assigned weight zero. 
 
Recall that a cycle separator is \emph{good} if it does not go through  any hole. 
We would like to separate $R$ using a good fundamental cycle separator $C_e$ of some edge $e$ w.r.t. \!$T_s$. If we can find such separator $C_e$ where $e$ is not incident to $v_h$ for some hole $h$, then $C_e$ is a good separator (since the vertices $v_h$ are leaves of the spanning tree $T_s$). Otherwise, we must separate $R$ with a fundamental cycle separator that goes through holes. We next define \emph{disposable} holes, and then show that we can allow the fundamental cycles to go through such holes.

Let $k$ be a node (level component) in $\mathcal K$. Let $C_k$ be the boundary cycle of $k$. Let $e$ be an edge of $C_k$. Note that $e \notin T_s$. This is because both endpoints $e$ have the same level, so, by Lemma~\ref{lem:triangulation}, neither can be the parent of the other in $T_s$.
Let $f,g$ be the endpoints of $e$ in the dual graph, such that $f$ is a face in $k$ and $g$ a face not in $k$.
Since $e \notin T_s$, $e$ is in the cotree $T_s^*$. Consider breaking $T_s^*$ into two subtrees by deleting $e$.
We say that the edge $e$ is {\em light} if the subtree of $T_s^*$ that contains $g$ has weight at most $W/2$ where $W$ is the total weight of the vertices of $R$. Note that we defined weights of primal vertices, whereas the vertices of $T_s^*$ are primal faces. To define face weights, evenly redistribute the weight of each vertex among all of its incident faces.
There is an equivalent, primal view of light edges:
The Jordan cycle $C_e$ partitions $R$ into two subgraphs, exactly one of which contains the faces of the level component corresponding to $k$. We say $e$ is \emph{light} if the weight of the subgraph that does not contain the level component $k$ is at most half the weight of $R$.  
We say that a level component $k$ is {\em disposable} in region $R$ if there are boundary edges of $k$ in $R$, and if every edge $e$ of the boundary of $k$ that is also in $R$ is light.
Note that, in particular, this definition applies to holes (since holes are level components). See Figure~\ref{fig:disposable-hole}.

\begin{figure}[h]
\begin{center}
\includegraphics[width=0.4\textwidth]{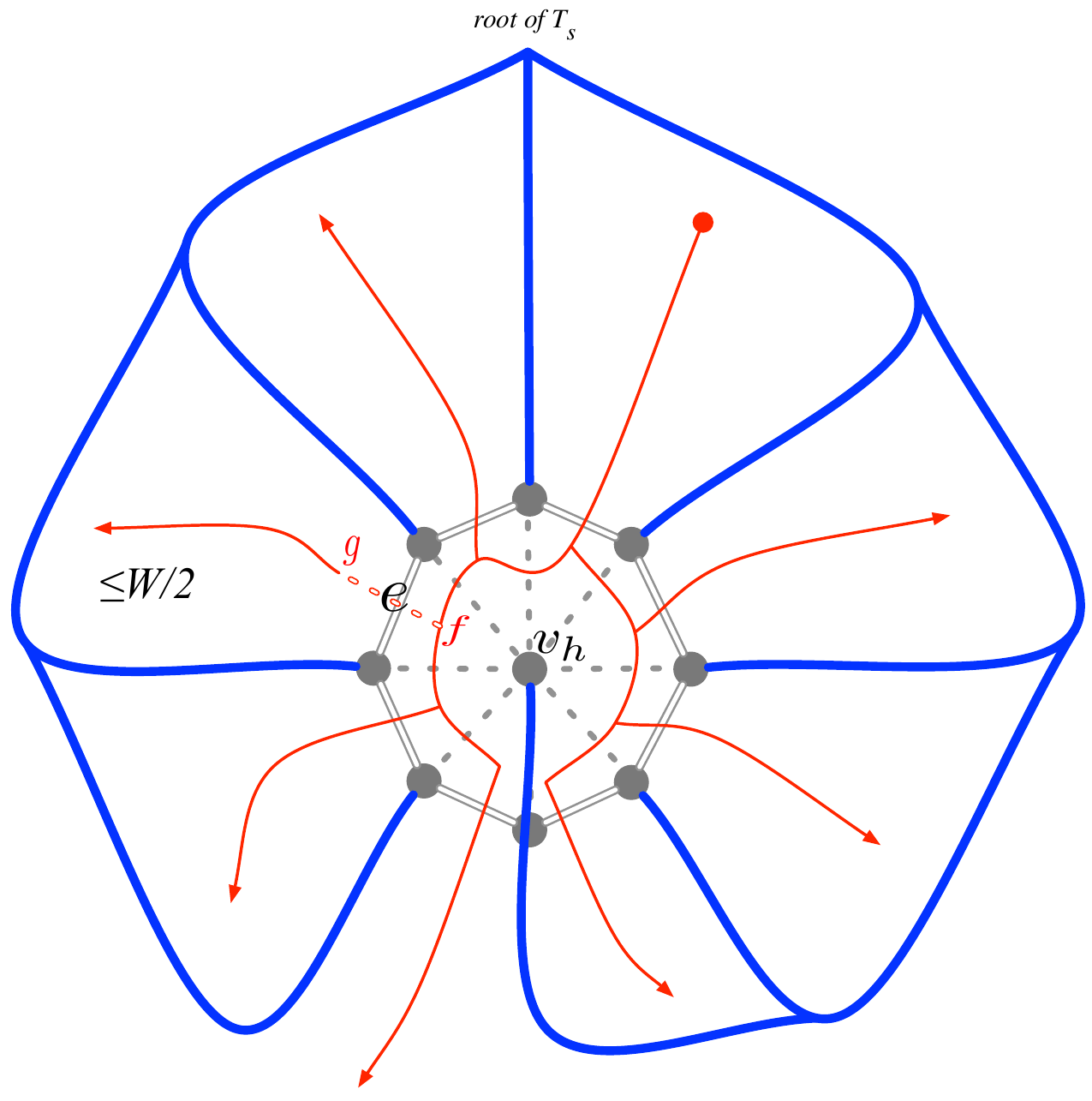}
\end{center}
\caption{An illustration of a region with a disposable hole. The edges of a boundary of a hole $h$ are shown in double-line grey. Since the boundary of a hole is a level cycle, none of the edges of the boundary of $h$ belongs to the spanning tree $T_s$ (blue). The artificial triangulation vertex $v_h$ of $h$ and the triangulation edges (grey dashed) are shown. The cotree $T^*_s$ is shown in thin red. Suppose that the number of distinguished nodes enclosed by $h$ is at least $W/2$ (so the weight of $v_h$ is at least $W/2$). Then, for any edge $e$ of the boundary of $h$, the part of $T^*_s \setminus e$ that does not contain faces of $h$ weighs at most $W/2$, so $e$ is a light edge and $h$ is  a disposable hole.  
\label{fig:disposable-hole}}
\end{figure}

Before showing why disposable holes exist and that they are useful, we first mention a simple property of $T^*_s$ and then use it to prove the existence of disposable holes.

\begin{property}
\label{dualtree} The cotree $T^*_s$ enters each level component exactly once. 
\end{property}
\begin{proof}
The spanning tree $T_s$ is monotone with respect to node levels. Thus, if $e$ is an edge of the boundary of a level component $k$, then one of the components of $T_s^* \setminus e$ contains no other faces, vertices or edges of $k$. See Figure~\ref{fig:tree-cotree} for an illustration.
\end{proof}

\begin{figure}[h]
\begin{center}
\includegraphics[width=0.5\textwidth]{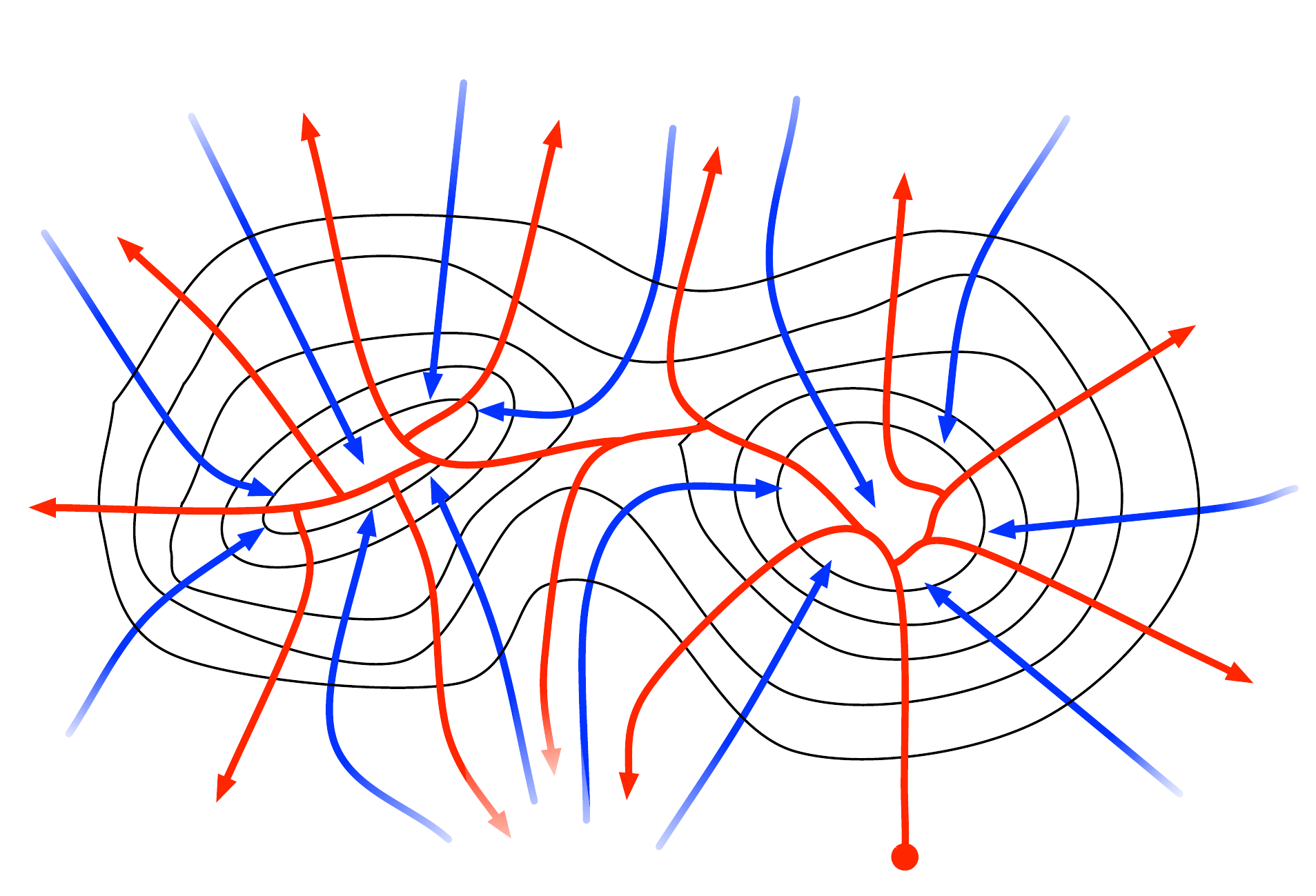}
\end{center}
\caption{An example of the interaction between the spanning tree $T_s$ (blue), the cotree $T_s^*$ (red), and boundary of level components (black cycles). Since $T_s$ is monotone with respect to levels, $T_s^*$ enters each level component exactly once. \label{fig:tree-cotree}}
\end{figure}

\begin{lemma}
\label{lem:recursion}
If a region $R$ contains more than one vertex with non-zero weight, then there exists
either a good balanced fundamental cycle separator or a disposable hole in $R$.
\end{lemma}

\begin{proof}
Let $W$ be the total weight of vertices in $R$. 
Consider the component tree $\mathcal K$. Let $u$ be a deepest disposable component in $\mathcal K$ such that $C_u$ has an edge in $R$.  If $u$ is a hole of $R$ then we found a disposable hole, and we are done. Otherwise, we next show that there exists a good separator.

Let $u_1,u_2,\ldots,u_d$ be the children of $u$ in $\mathcal K$ (if there is no disposable component in $R$, then define $u$ to be $R$, $C_u$ to be the external boundary of $R$, and let $u_1, \ldots, u_d$ be the set of rootmost components in $\mathcal K$  such that $C_{u_i}$ has an edge in  $R$).
Since none of the $u_i$'s is disposable, for each $u_i$ there exists exactly one boundary edge $e_i=(f_i,g_i)$ (here we wrote $e_i$ as a dual edge, and $f_i$ is the endpoint of $e_i$ that belongs to $u_i$), such that the subtree of $T_s^*\setminus e_i$ that contains $g_i$ has weight at least $W/2$. 
Consider the following two phase process (see Figure~\ref{fig:structural-lemma} for an illustration):  Let $T^*_0 = T_s^*$. If $T^*_i$ contains more than a single face of some $u_j$ (in which case it must contain all faces of $u_j$ by Property~\ref{dualtree}), then $T^*_{i+1}$ is obtained from $T^*_i$ by rooting $T^*_i$ at $g_j$ and deleting all the strict descendants of $f_j$ in $T^*_i$, so that $f_j$ becomes a leaf.
The weight assigned to $f_j$ in $T^*_{i+1}$ is the total weight of all the vertices in the deleted subtree. Thus, the weight of $T^*_{i+1}$ remains $W$, and, by definition of $e_j$, the weight of $f_j$ is at most $W/2$.
The first phase terminates when $T^*_i$ contains at most one face ($f_j$) from each $u_j$. In the second phase, while $T^*_i$ contains an edge $e$ of $C_u$ that is not a leaf edge of $T^*_i$, then $T^*_{i+1}$ 
is obtained from $T^*_i$ by rooting $T^*_i$ at the endpoint $g$ of $e$ that belongs to $u$, and deleting all the strict descendants of the other endpoint $f$ of $e$ in $T^*_i$, so that $f$ becomes a leaf.
Similarly to the first phase, the weight of $f$ in $T^*_{i+1}$ is set to the total weight of all the vertices in the deleted subtree. Since $u$ is disposable, the weight of $f$ is at most $W/2$.

Let $T^*_t$ be the resulting tree. Since $T^*_t$ contains at most one face from each $u_i$, $T^*_t$ contains no triangulation edges of a hole (both endpoints of a triangulation edge of a hole belong to the hole). Furthermore, 
the total weight of $T^*_t$ is $W$, and every leaf of $T^*_t$ created during the two phase process has weight at most $W/2$ (by definition). For the remaining nodes of $T^*_t$, the degree is at most 3 and the weight is also at most $W/2$, because the original weights in $T^*_s$ are at most $W/2$ (otherwise, the node corresponds to
a hole of weight at least $W/2$ that is, by definition, disposable, and we are done).
Therefore, there exists an edge $e$ whose deletion from $T^*_t$ results in two trees, none of which weighs more than $5W/6$.
By construction of the weights of $T^*_t$, the balance of the fundamental cycle of $e$ w.r.t. \!$T_s$ is exactly the ratio of the weights of the subtrees obtained by deleting $e$ from $T^*_t$.
Therefore, the fundamental cycle $C_e$ of $e$ w.r.t. \!$T_s$ is a balanced Jordan cycle separator. Since no edge of $T^*_t$ is a triangulation edge of a hole, $C_e$ is a good separator.
\end{proof}

\begin{figure}
\begin{center}
\includegraphics[width=0.3\textwidth]{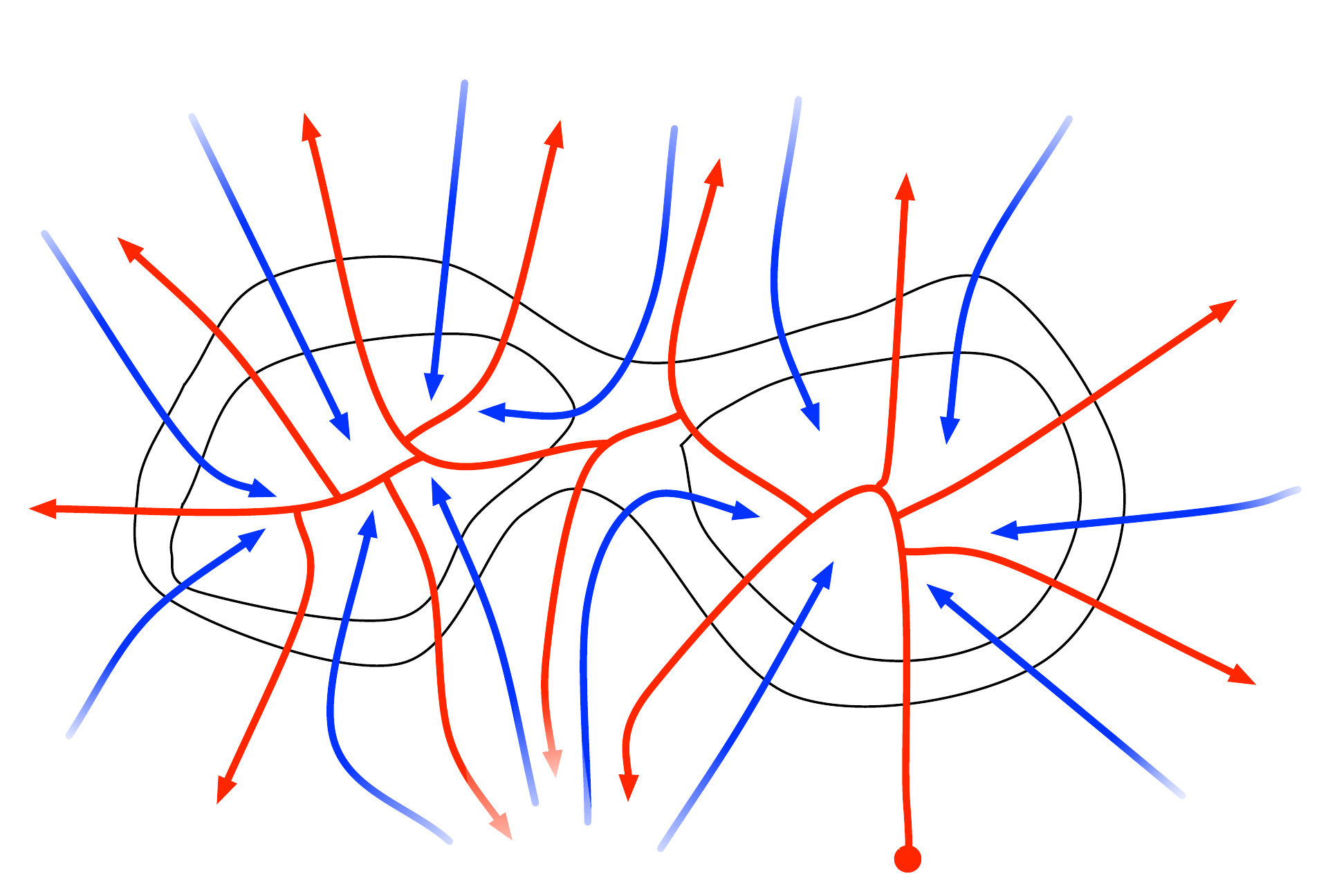}
\includegraphics[width=0.3\textwidth]{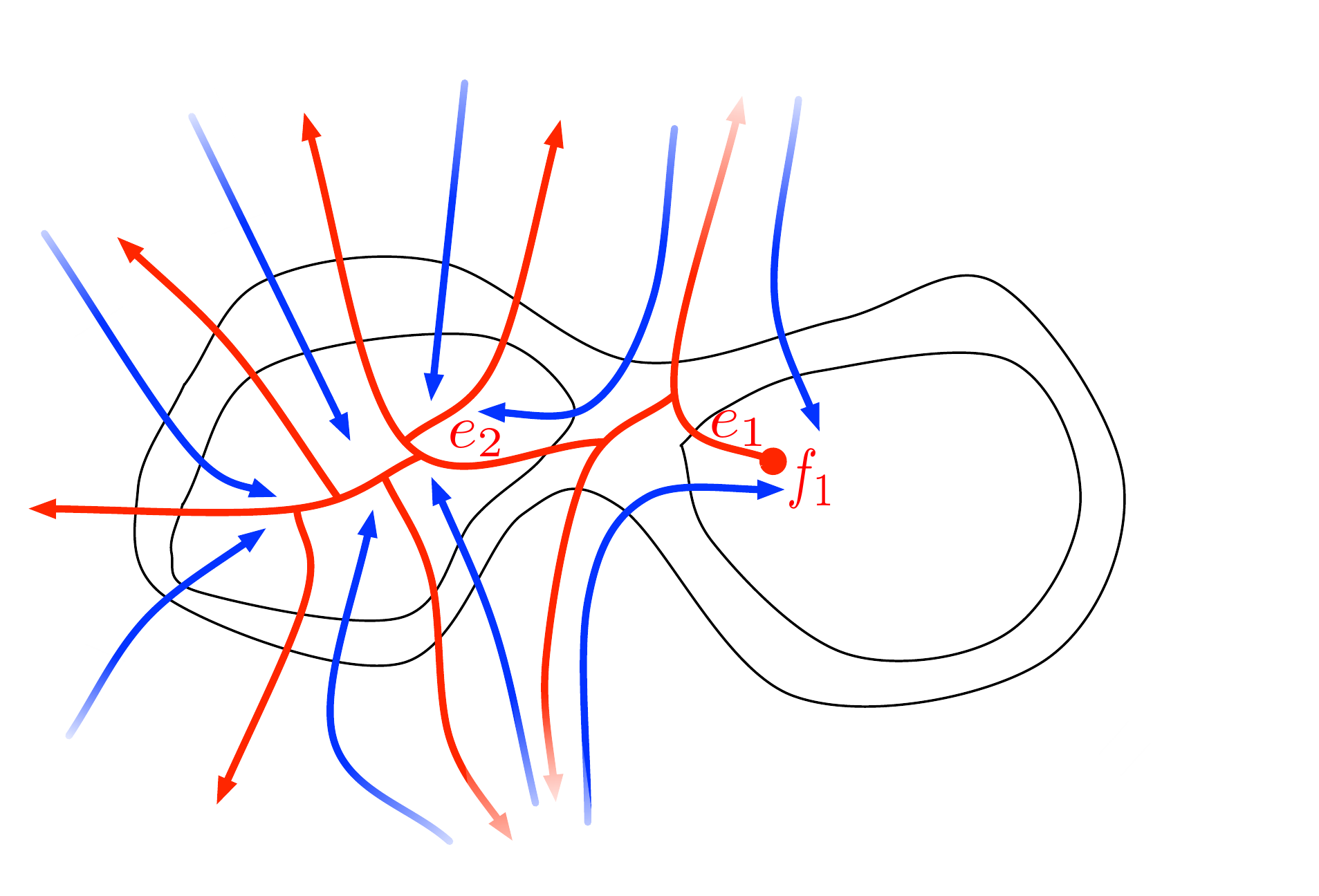}\\
\includegraphics[width=0.3\textwidth]{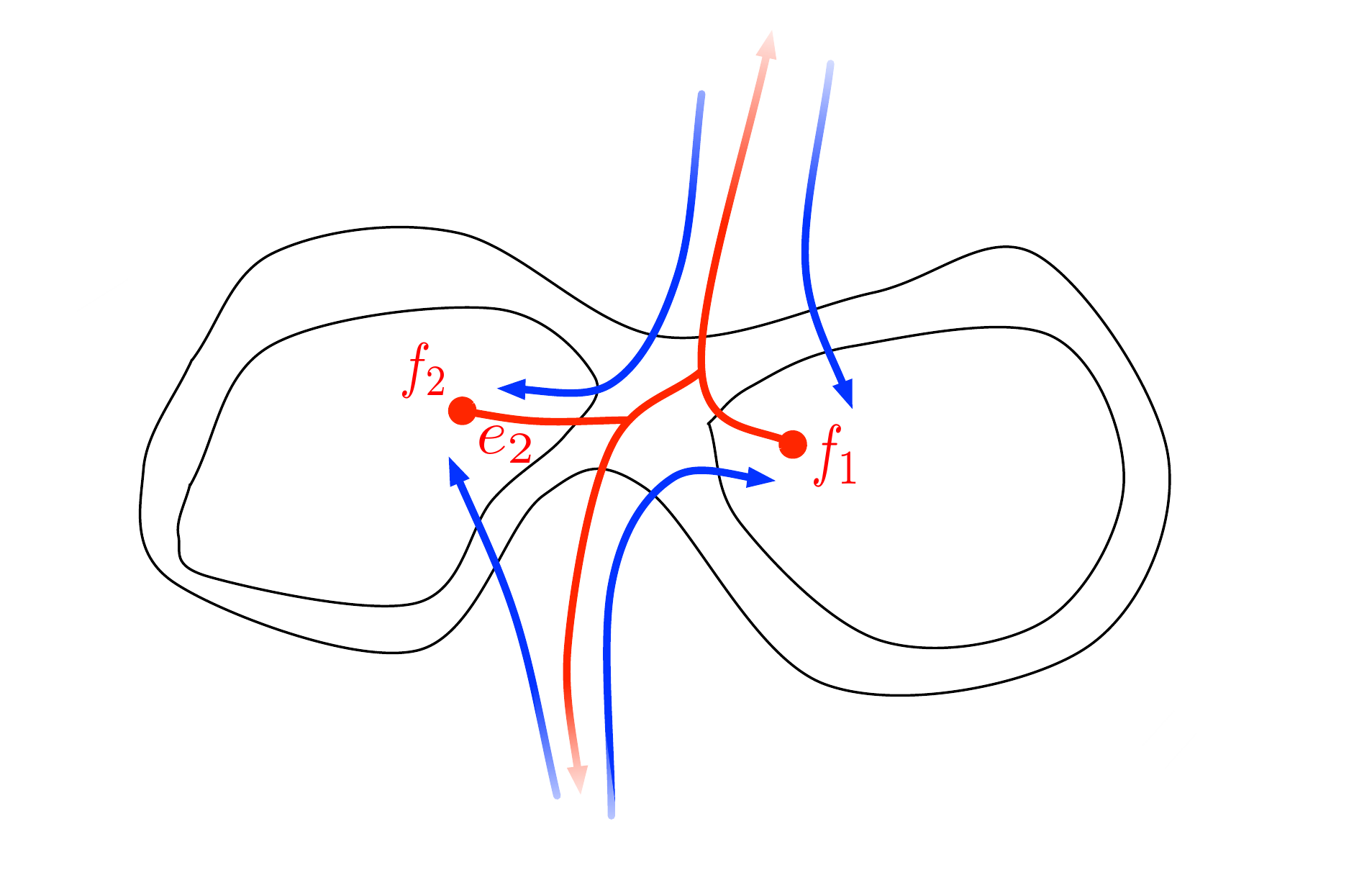}
\includegraphics[width=0.3\textwidth]{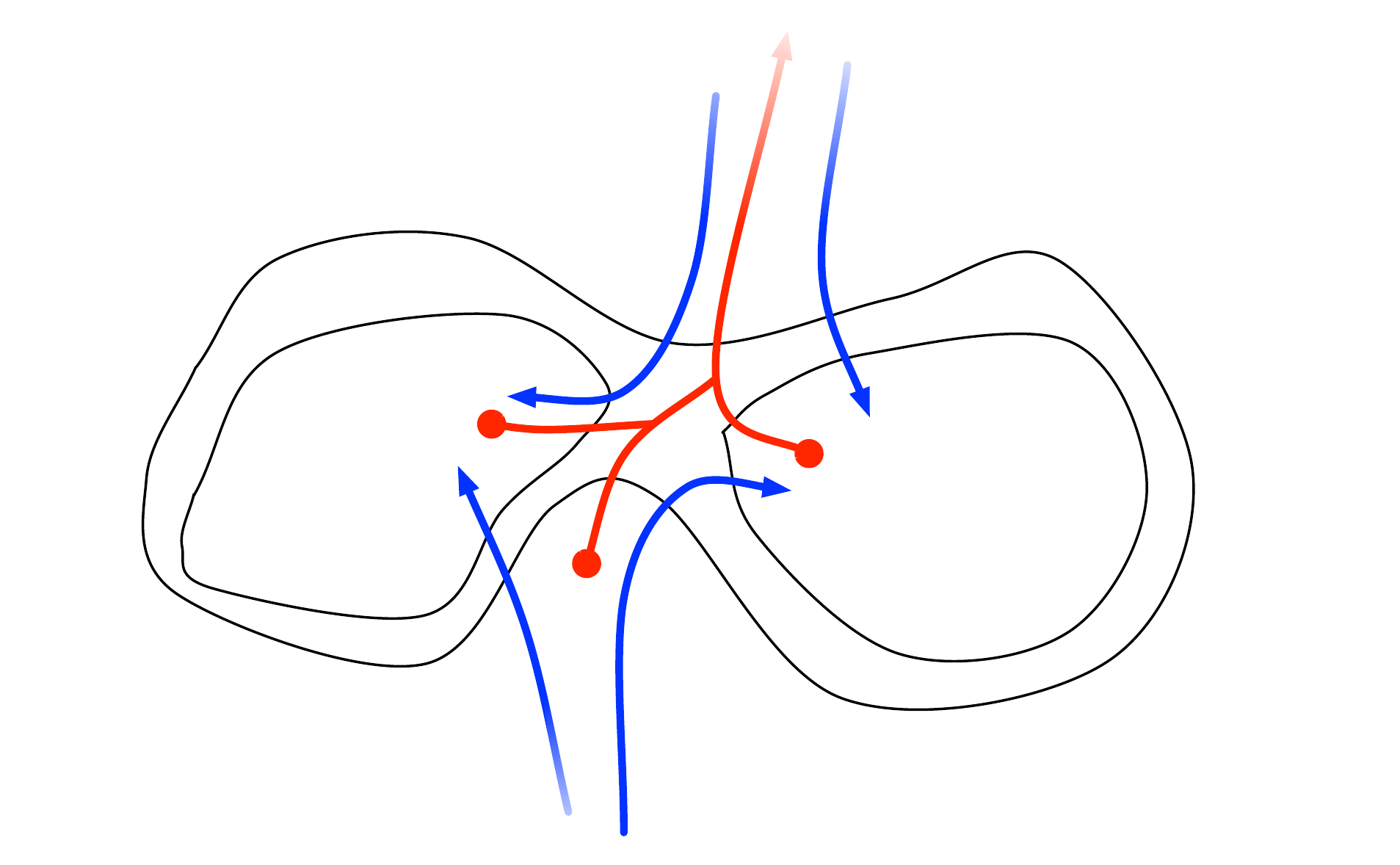}
\includegraphics[width=0.3\textwidth]{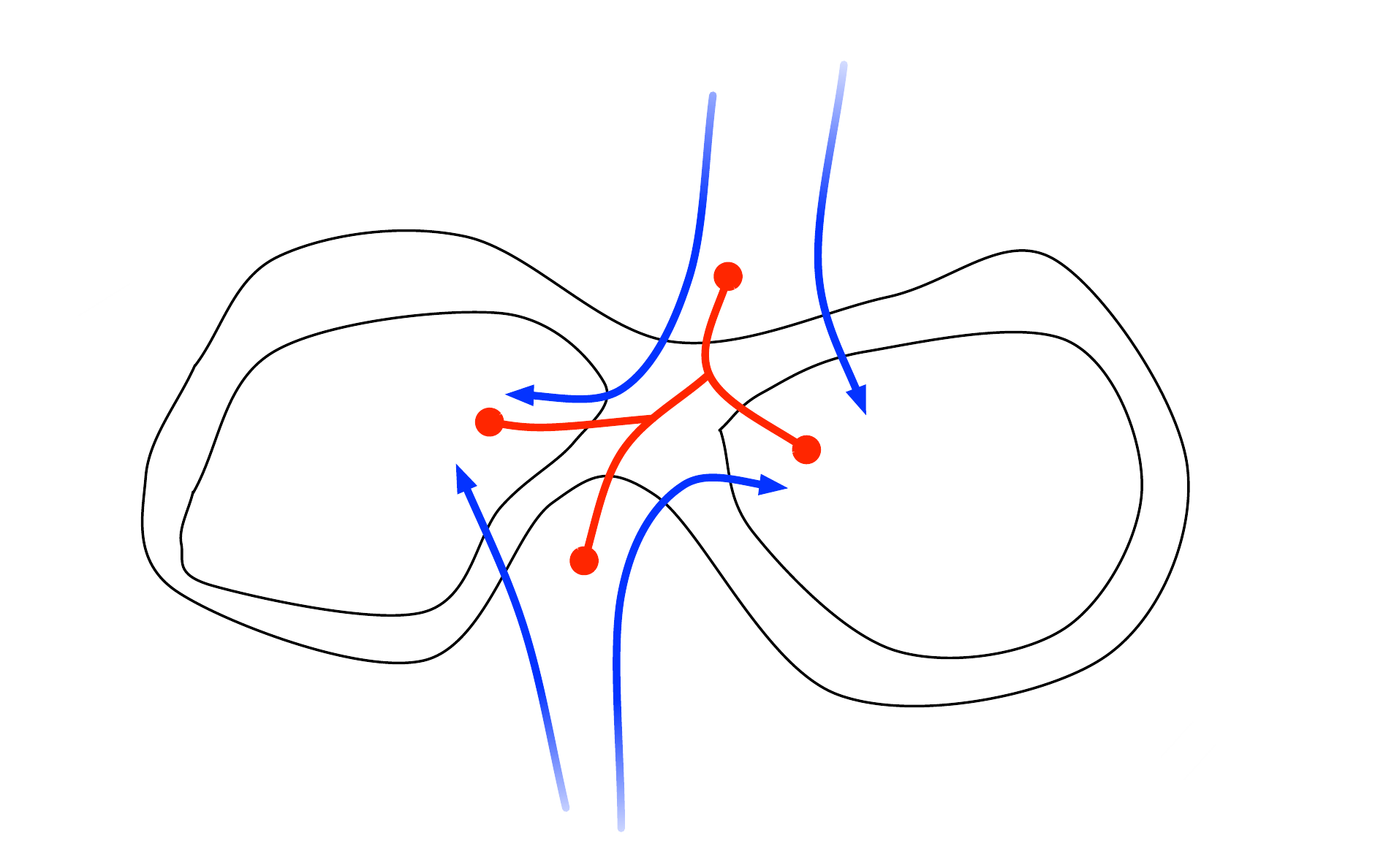}

\end{center}
\caption{Illustration of the process of constructing the cotree $T^*_t$ according to which a good separator is found. Three level boundaries are shown (black cycles). The external one is $C_u$, the two cycles enclosed by $C_u$ are $C_{u_1}$ and $C_{u_2}$. The subfigures show the process of splitting the cotree $T^*_s$, first at $e_1$, then at $e_2$. In the second phase the cotree is further split at the two remaining edges of $C_u$. The resulting tree $T^*_t$ is an induced subtree of $T^*_s$ in which a balanced edge-separator can be found. Since none of the edges of $T^*_t$ has both endpoints in any $u_i$, none of the edges of $T^*_t$ are triangulation edges of a hole. \label{fig:structural-lemma}}
\end{figure}

With this structural lemma we can now describe our oracle.
Consider a slice $s$ and
let $G_s$ be the subgraph of $G$ enclosed by the boundary of
$s$.
The goal of processing slice $s$ is to store information (distances) so that the following distances (in $G_s$) can be recovered from the information stored for all slices contained in $G_s$.
\begin{enumerate}
\item The distance between any two distinguished nodes in $G_s$,
\item The distance between any distinguished node in $G_s$ and any vertex on the
boundary of $s$.
\item The distance between any two vertices on the boundary of $s$.
\end{enumerate}
Encoding this information for all slices guarantees that distances between the
distinguished vertices in the whole graph are captured.

\paragraph{The encoding.}
To process a slice $s$, we first encode {\bf boundary-to-boundary} distances: the distances (in $G_s$) between
vertices on the boundary of $s$ using Lemma~\ref{lem:cycle}.
We then triangulate $s$ and define its spanning tree $T_s$ using Lemma~\ref{lem:triangulation}.

Next, we recursively separate $s$ using fundamental cycle separators.
The initial region $R$ is the entire slice $s$. Its boundary is the external boundary $C_{ext}$ of $s$. 
A region $R$ is separated into subregions obtained by cutting $R$ along some fundamental cycle separator $C$ w.r.t. \!$T_s$. 
Since we only use fundamental cycle separators w.r.t. \!the same tree $T_s$, the separators never cross. Hence, the boundary of each new region $R'$ consists of the contiguous portion of $C$ that belongs to $R$, and possibly portions of the boundary of $R$. 
Since $C$ crosses $C_{ext}$ at most twice (at most once for each of the two paths in the fundamental cycle $C$), the number of contiguous maximal fragments of $C_{ext}$ in the boundary of $R'$ is at most one plus the number of such fragments in the boundary of $R$. Consequently, the number of contiguous maximal fragments of $C_{ext}$ in the boundary of any region is bounded by the depth of the recursion, which we will show is $\tilde \Oh (1)$.

We now explain how to choose the fundamental cycle separator $C$ with which we separate $R$. This is achieved using two interleaving recursive processes. We refer to the first one as the outer recursion, and to the second one as the hole elimination recursion. 
In a step of the outer recursion we apply Lemma~\ref{lem:recursion}.
  
If we find a good balanced fundamental cycle separator $C$, then we use it to separate the  region $R$. Every vertex in $S^R$ explicitly stores {\bf $S$-to-separator} distances: its distance (in $G_s$) to every vertex of $C$. 
In addition, for each subregion $R'$, for each contiguous maximal fragment $b_i$ of $C_{ext}$ in $R'$, we encode {\bf separator-to-boundary} distances: the distances (in $R'$) between $b_i$ and $C$ using Lemma~\ref{lem:unit} or Lemma~\ref{lem:hole} (depending on whether the vertices of the separator $C$ and the vertices of $b_i$ lie on a single or two faces of $R'$). Then, we call the outer recursion recursively for each subregion $R'$.
The outer recursion terminates when there is at most one vertex with positive weight in the
current region $R$. 
If the only remaining object is an artificial vertex $v_h$, we apply Lemma~\ref{lem:hole} to
encode {\bf hole-to-boundary} distances: the distances (in $R$) between the boundary $C_h$ of $h$ and $b_i$, for each contiguous maximal fragment $b_i$ of $C_{ext}$ in $R$.
If the only remaining object is a distinguished vertex $u$,
we store {\bf $S$-to-boundary} distances: the distances (in $G_s$) from $u$ to every vertex of every $b_i$. If the current region $R$ contains no vertices with positive weight, the outer recursion terminates.

If, on the other hand, we found a disposable hole $h$, we store {\bf hole-to-boundary} distances: distances between the boundary $C_h$ of $h$ and every contiguous maximal fragment $b_i$ of $C_{ext}$ in $R$. 
The weight of the artificial vertex $v_h$ is set to zero. This reflects the fact that for the rest of the processing of $s$, distinguished vertices enclosed by the hole $h$ will not be treated individually and directly, but rather by encoding distances involving the vertices of $C_h$. From this point on, vertices of $S$ inside $h$ are no longer considered vertices of $S^R$. 
We then call the hole elimination process for the hole $h$ in region $R$ (see Figure~\ref{fig:disposable}).
In a single step of the hole elimination recursion, a region $R$ is separated using a fundamental cycle separator $C$ w.r.t. \!$T_s$ that is balanced w.r.t. \!the number of vertices of $C_h$ in $R$ (i.e., a weight 1 is assigned to each vertex of $C_h$ and 0 to all other vertices). Note that $C$ is necessarily a fundamental cycle w.r.t. \!$T_s$ of some triangulation edge that is incident to $v_h$. The boundary of each of the two resulting regions contains a single contiguous portion of $C_h$ consisting of roughly half the vertices of $C_h$ in $R$. Similarly to the single hole case, we store {\bf $S$-to-separator} distances: distances (in $G_s$) from every vertex of $S^R$ to every vertex of $C$. For each subregion $R'$ obtained by separating $R$ along $C$, for each contiguous fragment $b_i$ of $C_{ext}$ in $R'$, we encode {\bf separator-to-boundary} distances: the distances (in $R'$) between $b_i$ and $C$  using Lemma~\ref{lem:unit} or Lemma~\ref{lem:hole}, and {\bf separator-to-hole} distances: the distances (in $R'$) between $C$ and the single contiguous fragment of $C_h$ that belongs to $R'$ using Lemma~\ref{lem:unit}. 
We then apply the hole elimination process recursively to each subregion $R'$. It terminates when the current region $R$ contains at most two consecutive vertices of $C_h$, or when it contains at most one distinguished vertices. When this happens, we continue with the outer recursion on $R$.

\begin{figure}
\begin{center}
\includegraphics[width=0.25\textwidth]{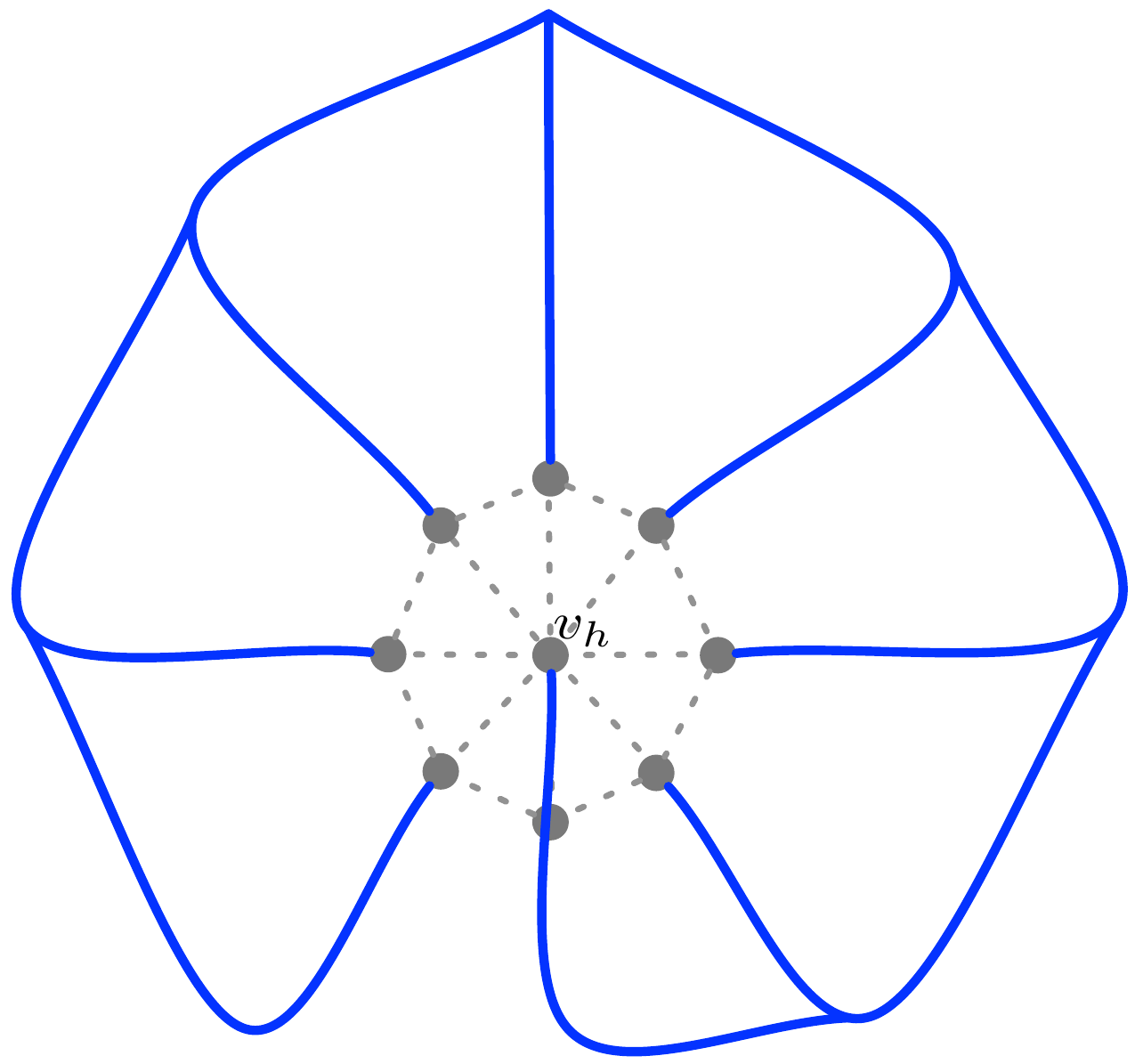}
\includegraphics[width=0.25\textwidth]{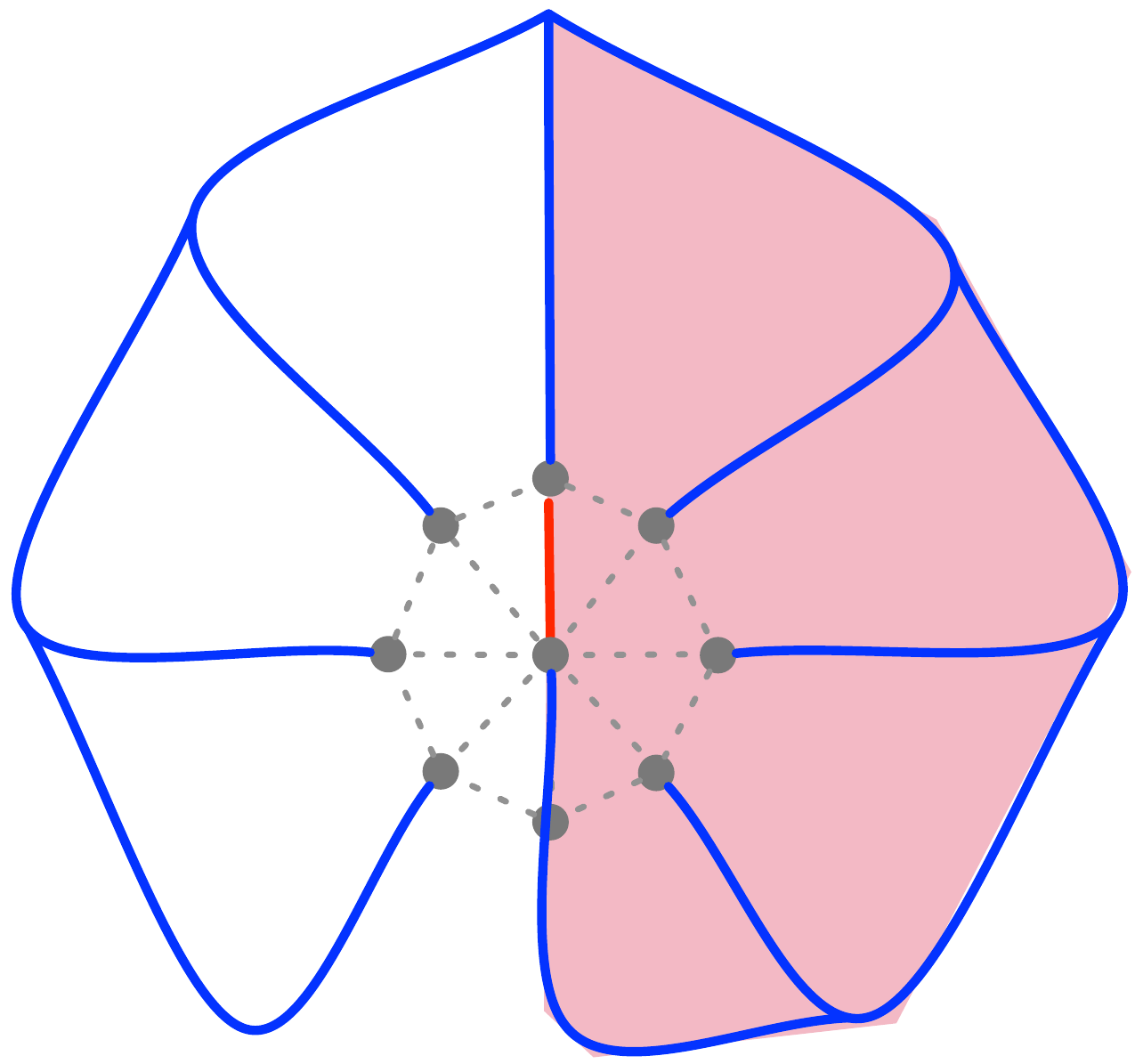}
\includegraphics[width=0.25\textwidth]{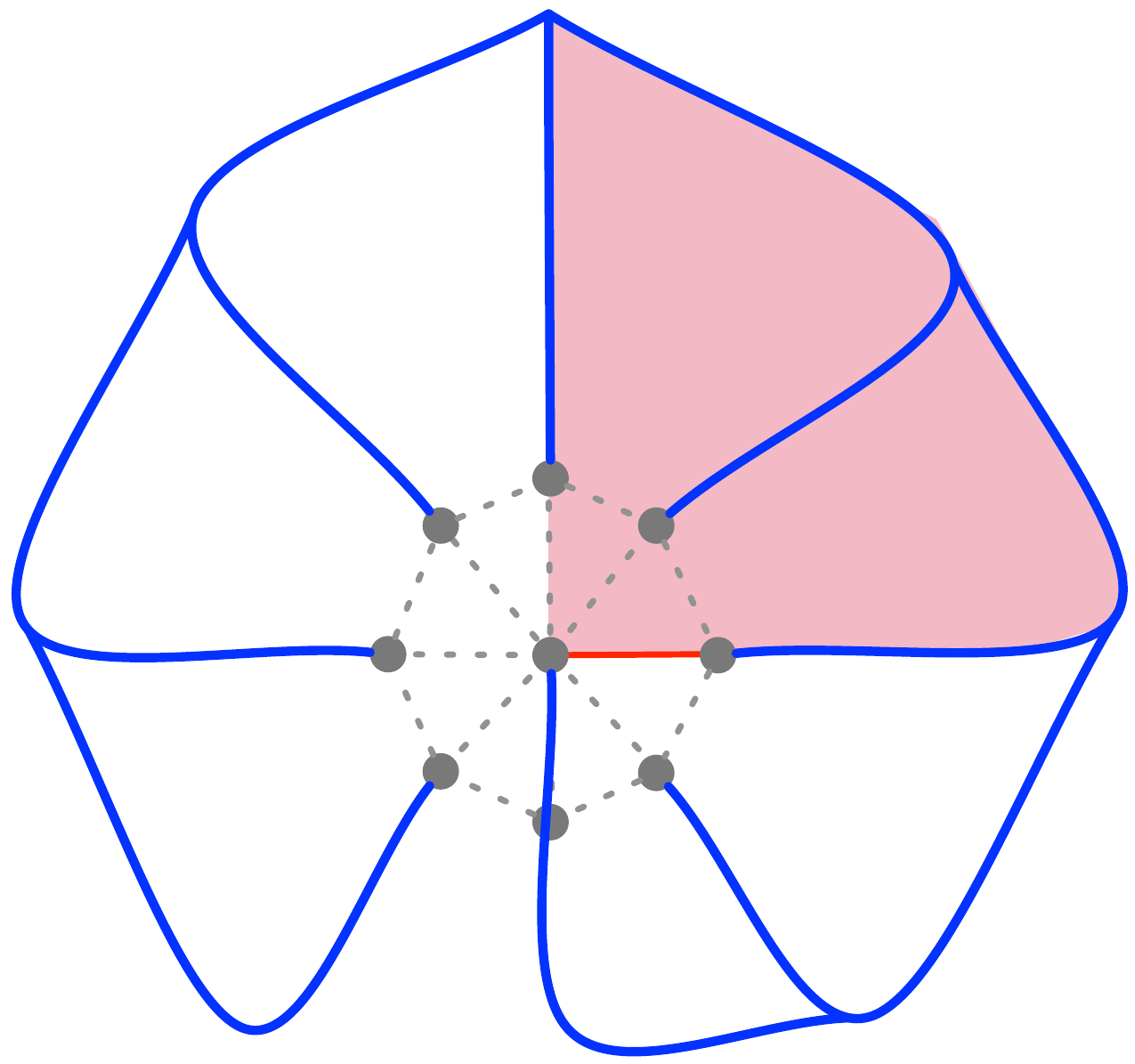}
\includegraphics[width=0.25\textwidth]{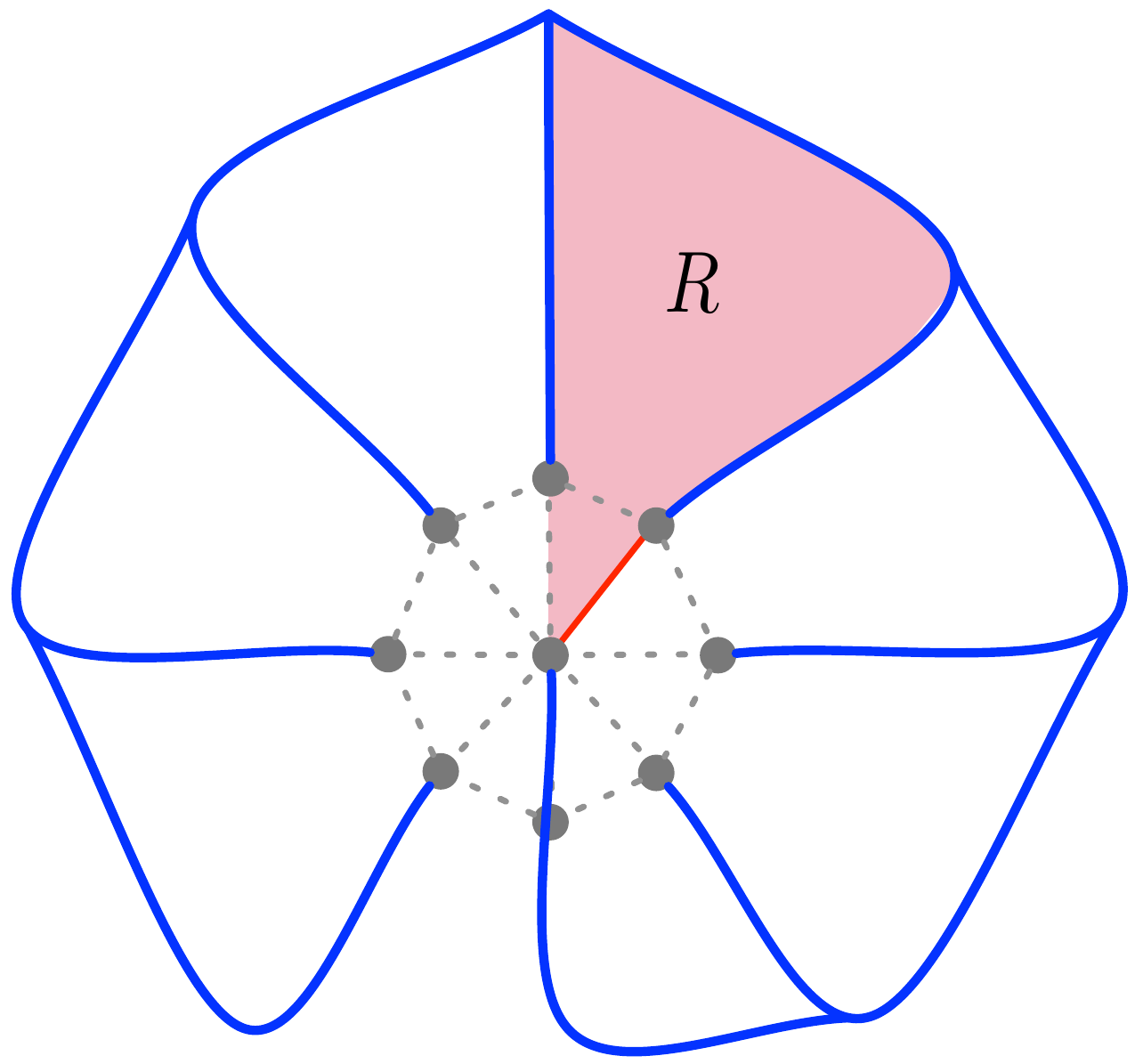}
\includegraphics[width=0.25\textwidth]{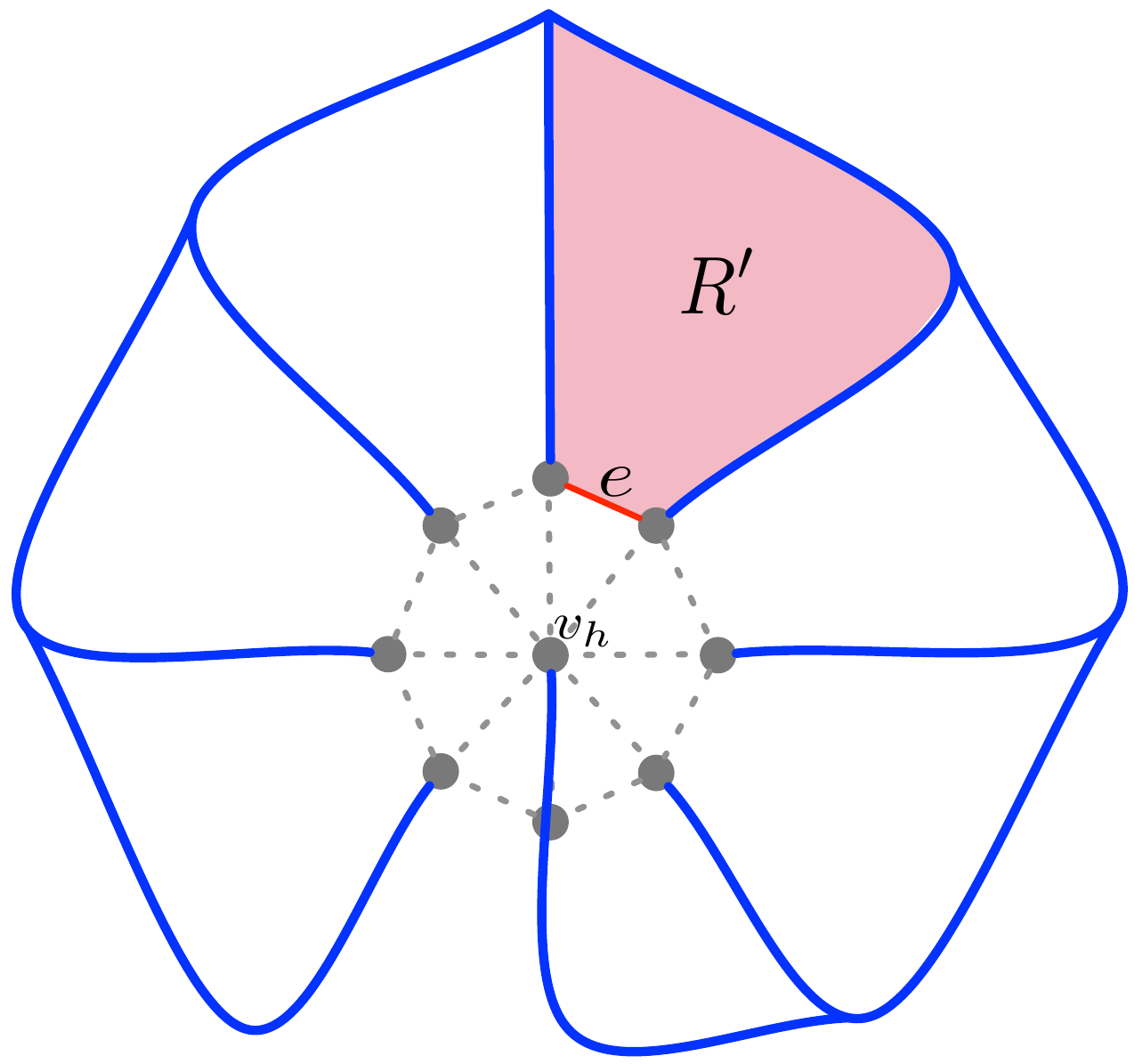}
\end{center}
\caption{Illustration of the process of eliminating a disposable hole. A disposable hole $h$ with artificial vertex $v_h$ in a region $R_0$ is shown in top-left. The spanning tree $T_s$ is indicated in blue. $R_0$ is recursively separated using fundamental cycle separators that are balanced w.r.t. \!the number of nodes of the boundary $C_h$ of $h$. The elimination process finishes (bottom-left) when the current region $R$ contains only two consecutive vertices of $C_h$, which are the endpoints of some edge $e$. This region $R$ differs from the the region $R'$ obtained by separating $R_0$ using the fundamental cycle of $e$ w.r.t. \!$T_s$ by a single vertex ($v_h$).\label{fig:disposable}}
\end{figure}

We next prove that the total depth of the entire recursive procedure is $\Oh(\log^2 n) = \tilde \Oh(1)$. 

\paragraph{Analyzing the recursion depth.} We begin with the initial region $R_0$ being the entire slice.  
In a single step of the outer recursion, if we find a good separator then we use it to separate the current region $R_0$ thus decreasing the weight of each resulting region $R$ by a constant factor. If however we do not find a good separator, then we apply the hole elimination process on a disposable hole $h$ of the current region $R_0$. Since $|C_h| = \Oh(n)$, and since every recursive call to the hole elimination process decreases the number of nodes of $C_h$ by half, we get that after $\Oh(\log n)$ recursive calls the hole elimination process terminates, with each resulting region $R$ containing only two nodes of $C_h$. Observe that these two nodes must be adjacent on $C_h$ (see Figure~\ref{fig:disposable}). Let $e$ be the edge between them and let $C_e$ be the fundamental cycle of $e$  w.r.t. \!$T_s$. Since $h$ is disposable, the weight of the region $R'$ obtained by separating $R_0$ using $C_e$ is at most half the weight of region $R_0$. Since $R = R' \cup \{v_h\}$ and since the weight of $v_h$ is zero this means that the weight of $R$ is at most half the weight of $R_0$. 
We conclude that every $\Oh(\log n)$ consecutive recursive calls the total weight of a region decreases by a constant factor. This shows that the depth of the recursion is $\Oh(\log^2 n)$.

\paragraph{Correctness.}
We next prove that the distance between any two distinguished vertices in $S$ can be recovered from our encoding.

\begin{lemma} \label{lem:x_s-to-B}
The length of a shortest path $P$ in $G_s$ from any  $x \in S_s$ to any  $y \in C_{ext}$ can be recovered from the encoding.
\end{lemma}
\begin{proof}
If $P$ contains some vertex of a fundamental cycle separator used in processing $s$, let $v$ be the last vertex of $P$ that belongs to the earliest such separator. By choice of the earliest separator, the $x$-to-$v$ distance (in $G_s$) is stored ($S$-to-separator distance). By choice of the last vertex on $P$ that belongs to that separator, the $v$-to-$y$ distance (in the region of $s$ that contains $P[v,y]$) is stored (separator-to-boundary distance). Thus, the length of $P$ can be recovered.

If $P$ contains no such vertex, then $x$ and $y$ are in the same region when the recursion terminates, so the $x$-to-$y$ distance (in $G_s$) is stored as a $S$-to-boundary distance. 
\end{proof}

\noindent We extend the previous lemma and show that it applies also to distinguished vertices enclosed by holes of $s$ (i.e., for $S^s$ instead of $S_s$).

\begin{lemma}\label{lem:x^s-to-B}
The length of a shortest path $P$ in $G_s$ from any  $x \in S^s$ to any  $y \in C_{ext}$ can be recovered from the encoding.	
\end{lemma}
\begin{proof}
The proof is by induction on the nesting depth of slice $s$. The base case follows from Lemma~\ref{lem:x_s-to-B}. 	
For the inductive step, if $x \in S_s$ we are done by Lemma~\ref{lem:x_s-to-B}, so assume $x$ is enclosed by some hole $h$ of $s$.

If $P$ contains some vertex of a fundamental cycle separator used in processing $s$ before hole $h$ is eliminated, let $v$ be the last vertex of $P$ that belongs to the earliest such separator. By choice of the earliest separator, the $x$-to-$v$ distance (in $G_s$) is stored ($S$-to-separator distance), and by choice of the last vertex of that separator on $P$, the $v$-to-$y$ distance (in a region of $s$ that contains $P[v,y]$) is stored (separator-to-boundary distance). Thus, the length of $P$ can be recovered.

If $P$ contains no such vertex, then the artificial vertex $v_h$ and $y$ are in the same region $R$ when either the recursion terminates, or the hole $h$ is eliminated. In either case, the $C_h$-to-$y$ distances are stored (hole-to-boundary distance). Decompose $P$ into a maximal prefix $P[x,v]$ enclosed by the slice $s'$ whose boundary is $C_h$, a maximal suffix $P[w,y]$ enclosed by $R$, and an infix $P[v,w]$.
The length of the prefix is stored by the inductive hypothesis for $s'$. The length of the infix is represented by the boundary-to-boundary distances for $s'$. The length of the suffix is stored (hole-to-boundary distance).
\end{proof}

\noindent Finally, we extend the previous lemma and show that it applies to any two distinguished vertices.

\begin{figure}
\begin{center}
\includegraphics[width=0.5\textwidth]{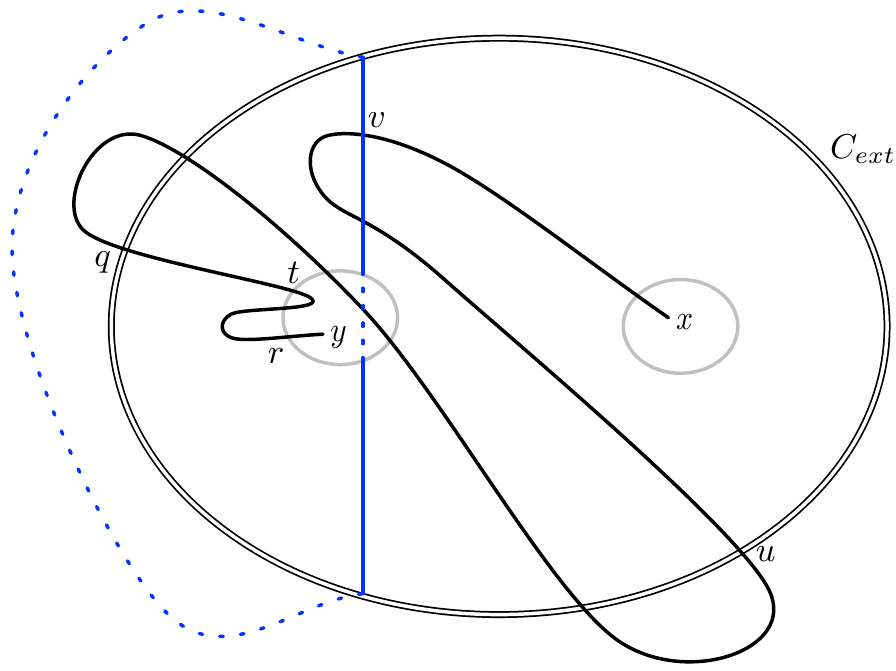}
\end{center}
\caption{An illustration of one of the cases in the proof of Lemma~\ref{lem:x-to-x}. A slice $s$ with two holes is shown. The boundary $C_{ext}$ is double-lined. A $x$-to-$y$ shortest path $P$ is shown in solid black. The vertices $x$ and $y$ belong to different holes of $s$ (black cycles). The path $P$ crosses a fundamental cycle separator (blue, parts that do not belong to $s$ are dashed) used in eliminating the hole to which $y$ belongs. \label{fig:correct}}
\end{figure}

\begin{lemma}\label{lem:x-to-x}
The length of a shortest path $P$ in $G_s$ from any  $x\in S^s$ to any $y\in S^s$ can be recovered from the encoding.	
\end{lemma}
\begin{proof}	
Assume, wlog, that both $x$ and $y$ are enclosed in holes of $s$ (the other cases are similar and less general). 
If $P$ contains some vertex of a fundamental cycle separator used in processing $s$ before either hole is eliminated, then let $v$ be a vertex of $P$ that belongs to the earliest such separator. By choice of the earliest separator, both the $x$-to-$v$ and the $y$-to-$v$ distances (in $G_s$) are stored ($S$-to-separator distance). Thus, the length of $P$ can be recovered.

Otherwise, the hole of $x$ and the hole of $y$ are in the same region $R$ when one of them, say the hole $h$ of $x$, is eliminated. If $P$ intersects one of the fundamental cycle separators used during the elimination process of hole $h$, then let $v$ be the last vertex on the earliest such separator (see Figure~\ref{fig:correct}). By choice of earliest separator, the $x$-to-$v$ distance is stored ($S$-to-separator distance). By Lemma~\ref{lem:x^s-to-B}, the length of the maximal suffix $P[w,y]$ enclosed in $h$ is also stored. Let $u$ be the first vertex of $P[v,w]$ that belongs to either $C_{ext}$ or $C_h$ ($u$ exists because $y$ is enclosed by $C_h$ and $v$ is not). 
\begin{itemize}
\item
If $u$ belongs to $C_{ext}$ then the length of $P[v,u]$ is stored (separator-to-boundary distance). In this case, let $q$ be the last vertex of $P$ that belongs to $C_{ext}$. The length of $P[u,q]$ is represented by boundary-to-boundary distances for $C_{ext}$. Let $t$ be the first vertex of $P[q,y]$ that belongs to $C_h$. The length of $P[q,t]$ is represented as a hole-to-boundary distance (when $h$ is eliminated). Let $r$ be the last vertex of $P[t,y]$ that belongs to $C_h$. The length of $P[t,r]$ is represented by boundary-to-boundary distances for $C_h$, and the length of $P[r,y]$ is represented by Lemma~\ref{lem:x^s-to-B}. See Figure~\ref{fig:correct} for an illustration.
\item
If $u$ belongs to $C_h$ then the length of $P[v,u]$ is represented as a separator-to-hole distance. The representation of the suffix $P[u,y]$ is then similar to the previous case.
\end{itemize}

Finally, we need to treat the case where $P$ does not intersect any fundamental cycle used in eliminating the hole $h$. In this case $P$ can be decomposed into a $x$-to-$C_{ext}$ prefix, a $C_h$-to-$y$ suffix, and subpaths of $P$ between vertices of $C_h \cup C_{ext}$. The prefix and suffix are represented by Lemma~\ref{lem:x^s-to-B}. The other subpaths are represented as hole-to-boundary or boundary-to-boundary distances as in the two cases above.
\end{proof}

\noindent Finally, we now show that the entire encoding requires only $\tilde\Oh(\sqrt{k\cdot n})$ bits.

\paragraph{The encoding size.}

The space required for the boundary-to-boundary distances for all slices is $O(n/w)$ since the total boundary size is $\Oh(n/w)$, and by Lemma~\ref{lem:cycle}.

We next bound the total space required for $S$-to-separator distances for all slices.
Whenever a distinguished vertex stores its distances to a path $P$
explicitly, the total weight of its region decreases by a constant factor within $\tilde \Oh(1)$ recursive steps (either immediately, if this happens in the outer recursion, or otherwise by the time the hole-elimination process ends). So this can happen $\tilde \Oh(1)$ times per distinguished vertex.
Because $|P|=\Oh(w)$ (by the
height of $T_s$), this sums up
to a total of $\tilde\Oh(k\cdot w)$ bits.

The analysis of the remaining distances is done for each slice separately. We have already argued that the depth of the recursive process to handle a slice $s$ is $\Oh(\log^2 n) = \tilde \Oh(1)$. Similarly to the analysis in Section~\ref{singleholesection} of the single hole case, the total space required for storing separator-to-boundary distances using Lemma~\ref{lem:unit} or Lemma~\ref{lem:hole} at all calls at the same recursive level is $\Oh(n/w + k + k\cdot w)$. For exactly the same reasons, the total space required for storing separator-to-hole distances using Lemma~\ref{lem:unit} at all calls at the same recursive level (this only happens in the hole-elimination recursion) is $\Oh(n/w +  k + k\cdot w)$. 

Hole-to-boundary distances are stored using Lemma~\ref{lem:hole} for at most one hole in each region along the recursion. Each invocation of Lemma~\ref{lem:hole} for hole $C_h$ and boundary fragment $b_i$ requires $\tilde \Oh(|C_h| + |b_i|)$ bits. For a single level of the recursion, this sums up to $\tilde \Oh(k + n/w)$ because the total size of all boundaries is $\Oh(n/w)$ and there are $\Oh(k)$ vertices that contribute in more than one region (endpoints of $b_i$'s). The bound for $S$-to-boundary distances is $\tilde \Oh(k + n/w)$ for the same reasons.

To conclude, we showed that the total size of the entire encoding is bounded by $\tilde\Oh(n/w+k\cdot w)$, which is $\tilde\Oh(\sqrt{k\cdot n})$ by choosing $w=\sqrt{n/k}$.

\section{A Tight Lower Bound}
\label{sec:lowerbound}

Recall that Gavoille et al.~\cite{GPPR04} show how to construct,
given a Boolean
$\frac{k}{2}\times \frac{k}{2}$ matrix $B$, a planar grid $G(B)$ containing
$\Oh(k^3)$ vertices, such that $B$
can be recovered from the distances between $k$ distinguished vertices of
$G(B)$. This shows that, for
$k\leq n^{1/3}$, encoding all distances between $k$ vertices of a planar graph
requires
$\Theta(k^2)$ bits. For $k\geq n^{1/3}$, we consider $t$ Boolean $\frac{k}{2t}
\times \frac{k}{2t}$
matrices $B_1,B_2,\ldots, B_t$. For each of these matrices, we construct a planar
grid containing
$\Oh((\frac{k}{2t})^3)$ vertices. The disjoint union of all these grids is a
planar graph
on $\Oh(t(\frac{k}{2t})^3)=\Oh(k^3/t^2)$ vertices, such that all Boolean
matrices can be recovered from
the distances between the $k$ distinguished vertices. Hence, encoding all such
distances requires
$\Omega(t(\frac{k}{2t})^2)=\Omega(k^2/t)$ bits. Setting $t=\sqrt{k^3/n}$ we
obtain that
encoding all distances between the $k$ distinguished vertices of a planar graph on
$n$ vertices
requires $\Omega(\sqrt{k\cdot n})$ bits.

\section{Query Time}
\label{sec:query}

The goal of Section~\ref{sec:upper} was to guarantee that all distances between
distinguished vertices are captured, but we were not concerned with the complexity
of retrieving such a distance. In this section we explain how to augment the encoding
to allow efficient extraction of the stored distances.

We start with reformulating our encoding using the notion of \emph{dense distance graphs}.
Vertices of a dense distance graph are listed explicitly, but its edges are described implicitly
with unit Monge matrices. Each such matrix describes lengths of the edges between
every $u\in U$ and $v\in V$, for some subsets of nodes $U$ and $V$.
The matrix is represented using $\tilde\Oh(|U|+|V|)$ bits as described in Lemma~\ref{lem:unit}.
In particular, we may have $|U|=|V|=1$ and then the matrix simply stores the length
of a single edge explicitly. The size of a dense distance graph is the total number of vertices
plus the sum of $|U|+|V|$ over all matrices describing length of the edges.
By construction, our encoding described in Section~\ref{sec:upper} is based
on defining a dense distance graph of size $\tilde\Oh(\sqrt{k\cdot n})$. Every distinguished
node of the original graph is a vertex of the dense distance graph, and the distance between
two distinguished nodes of the original (unweighted) graph is the same as the distance
between their corresponding vertices in the (weighted) dene distance graph.
Fakcharoenphol and Rao designed an efficient algorithm for computing the shortest paths
in such a graph, nicknames the FR-Dijkstra:~\footnote{In the original paper and most of
the subsequent work, the dense distance graph is obtained from an $r$-division of a planar
graph. The vertices are the boundary nodes and distances between boundary nodes in the
same region are represented with multiple Monge matrices. However, it is easy to see that
their algorithm work for any dense distance graph as defined above.}

\begin{lemma}[\cite{FR06}]
\label{lem:FR}
Distance between any two vertices of a dense distance graph of size $s$ can be found in $\tilde\Oh(s)$
time.
\end{lemma}

Applying Lemma~\ref{lem:FR} gives us an oracle of size $\tilde\Oh(\sqrt{k\cdot n})$ answering
queries in $\tilde\Oh(\sqrt{k\cdot n})$ time. For very large $k$, say $k=\Omega(n)$, the query
time is clearly not optimal, as there exists an oracle of size $\tilde\Oh(n)$ answering queries
in $\tilde\Oh(\sqrt{n})$~\cite{FR06} time. In the remaining part of this section we will
describe how to construct an oracle of size $\tilde\Oh(\sqrt{k\cdot n})$ answering queries
in $\tilde\Oh(n^{3/4})$ time.

To improve the query time, we apply the vanilla planar separator lemma.

\begin{lemma}
\label{lem:separator}
For any planar graph $G$ on $n$ nodes, there exists a partition of the nodes of $G$ into sets
$A$, $B$, and $S$, such that $|A|,|B|\leq \frac{2}{3}n$, $|S|=\Oh(\sqrt{n})$, and there are no
edges between the nodes of $A$ and $B$.
\end{lemma}

We recursively apply Lemma~\ref{lem:separator} to construct a hierarchical decomposition of the
whole graph. The recursion is described by a binary tree $\mathcal{K}$, where every node
$u\in\mathcal{K}$ corresponds to an induced subgraph $G(u)$ of the original graph. We let $n(u)$ and
$s(u)$ denote the number of nodes and distinguished nodes in $G(u)$, respectively. We terminate
the recursion as soon as $s(u) \leq \sqrt{n(u)}$. If $u\in\mathcal{K}$ is a leaf, we define its set of
distinguished nodes $D(u)$ to consists of all the distinguished nodes of $G(u)$. Otherwise, $D(u)$
consists of the following nodes:
\begin{enumerate}
\item the separator of $G(u)$,
\item for every child $v$ of $u$ that is a leaf, all the distinguished nodes of $G(v)$,
\item for every child $v$ of $u$ that is not a leaf, the separator of $G(v)$.
\end{enumerate}
Then, we construct a dense distance graph of size $\tilde\Oh(\sqrt{n(u)|D(u)|})$ capturing distances between
any two nodes from $D(u)$ in $G(u)$.

To calculate the distance $d_{G}(u,v)$ between two
distinguished nodes $u$ and $v$ in $G$, we locate the deepest nodes $u'$ and $v'$ of
$\mathcal{K}$, such that $u\in G(u')$ and $v\in G(v')$. Then, we consider the union of all dense distance graphs
constructed for the nodes of $\mathcal{K}$ on the paths from $u'$ and $v'$ to the root.
Note that the same node of $G$ might appear in more than one of these dense distance graphs, and we identify all
of its copies. By construction, the obtained dense distance graph captures the sought distance. Furthermore,
its size is bounded by
\[ \max_{u\in \mathcal K}\tilde\Oh(\sqrt{n(u)|D(u)|}))= \max_{u\in \mathcal{K}}\tilde\Oh(\sqrt{n(u)\sqrt{n(u)}})  = \tilde\Oh(n^{3/4}). \]
Therefore, by Lemma~\ref{lem:FR} we can answer a query in $\tilde\Oh(n^{3/4})$ time. It remains
to bound the size of the resulting oracle.

\begin{lemma}
\label{lem:nodesize}
The dense distance graph constructed for node $u\in\mathcal{K}$ is of size $\tilde\Oh(\sqrt{n(u)\cdot s(u)})$.
\end{lemma}

\begin{proof}
To prove the lemma it is enough to bound $|D(u)|$ by $\Oh(s(u))$. If $u$ is a leaf, this is clear.
Otherwise, $s(u) > \sqrt{n(u)}$ and $D(u)$ consists of the following nodes:
\begin{enumerate}
\item the separator of $G(u)$ of size $\Oh(\sqrt{n(u)}) = \Oh(s(u))$.
\item for every child $v$ of $u$ that is a leaf, all $s(v)\leq s(u)$ distinguished nodes of $G(v)$,
\item for every child $v$ of $u$ that is not a leaf, the separator of $G(v)$ of size $\Oh(\sqrt{n(v)})=\Oh(\sqrt{n(u)})=\Oh(s(u))$.
\end{enumerate}
Node $u$ has at most two children, so indeed $D(u)=\Oh(s(u))$.
\end{proof}

To upper bound the size of the oracle, we need to upper bound the sum
$\sum_{u\in \mathcal{K}} \sqrt{n(u)\cdot s(u)}$. To this end, we separately consider all nodes
$u\in\mathcal{K}$ such that $n(u)\in [(\frac{2}{3})^{\ell+1}n,(\frac{2}{3})^{\ell}n)$, for every $\ell=0,1,\ldots,\Oh(\log n)$.
Fix $\ell$ and call these nodes $u_{1},u_{2},\ldots,u_{t}$. Then, no $u_{i}$ is a descendant of another $u_{j}$,
so every node of the original graph appears in at most one $G(u_{i})$. Therefore, $\sum_{i} s(u_{i})\leq k$
and $\sum_{i} n(u_{i}) \leq n$. From the latter inequality and the lower bound on $n(u)$ we obtain
that $t \leq (\frac{3}{2})^{\ell+1}$. Now we want to upper bound the following sum:
\[ \sum_{i} \sqrt{n(u_{i}) \cdot s(u_{i})} \leq \sqrt{\left(\frac{2}{3}\right)^{\ell}n} \sum_{i} \sqrt{s(u_{i})} = \Oh(\sqrt{n/t} \cdot \sum_{i} \sqrt{s(u_{i})}). \]
From the concavity of $f(x)=\sqrt{x}$, the above sum is maximized when $s(u_{i})=k/t$, so we obtain:
\[ \sum_{i} \sqrt{n(u_{i}) \cdot s(u_{i})} = \Oh(\sqrt{n/t} \cdot t \cdot \sqrt{k/t}) = \Oh(\sqrt{k \cdot n}). \]
To obtain an upper bound on $\sum_{u\in \mathcal{K}} \sqrt{n(u)\cdot s(u)}$, we only
need to multiply the above bound by $\log n$ because for every $u\in\mathcal{K}$ there exists
$\ell$ such that $n(u)$ belongs to the appropriate interval, so the total size of the oracle is
$\tilde\Oh(\sqrt{k \cdot n})$.

\section{Labeling Schemes for Unit-Monge Matrices}
\label{sec:unit}

A distance labeling scheme is a way to compress graphs that allows for distributed decoding. The goal is to assign a label $\ell_v$ for each node $v$, so that by looking at the labels of two nodes $\ell_s,\ell_t$ (without access to the original graph) we can infer the distance between them $d(s,t)$.
The main question one asks about such schemes is \emph{how small can the labels be}? A famous open question is to close the gap between the $O(\sqrt{n})$ upper bound \cite{GPPR04,GawrychowskiU16} and the $\Omega(n^{1/3})$ \cite{GPPR04} lower bound for planar graphs.
The only known technique capable of proving a tight lower bound is via a lower bound for the metric compression problem: if you show that the metric cannot be compressed into $O(k \cdot n^{1/2-\eps})$ bits, then you show that no labels of size $O(n^{1/2-\eps})$ are possible.
Our work deems this approach impassable, since such compressions are indeed possible.
Optimistically, it is natural to ask if our upper bound for compression could lead to a better upper bound for labeling. 
Our encoding assigns $o(n^{1/2})$ bits per node, but can we distribute these bits to the nodes while allowing any pair of nodes to deduce the distance from their local information?
In this section, we discuss why this seems difficult.

The heart of our encoding is Lemma~\ref{lem:unit}, which is repeatedly used to capture pairwise distances between a large subset
of nodes of the graph using space proportional to the size of the subset. A key part in the proof of the lemma is
an efficient encoding of an $n \times n$ matrix into $\tilde{O}(n)$ bits, as long as it has the \emph{unit-Monge} property, that is:
\begin{align*}
M[i+1,j+1] +M[i,j]- M[i,j+1] - M[i+1,j] &\geq& 0& \qquad\text{ for any } i,j
\in [1,n-1], \\ 
|M[i,j] - M[i+1,j]| &\leq& 1 &\qquad\text{ for any } i\in [1,n-1], j\in [1,n],
\\ 
|M[i,j] - M[i,j+1]| &\leq& 1 &\qquad\text{ for any } i\in [1,n], j\in [1,n-1].
\end{align*}
and every $M[i,j]$ is an non-negative integer not exceeding $n$.
The corresponding labeling problem
would be to assign
a label to every row and every column of $M$, such that $M[i,j]$ can be
computed from the label
of the $i$-th row and the $j$-th column.
We will show that the $\tilde{O}(n)$ bits of the encoding \emph{cannot} be distributed into $\tilde{O}(1)$ bits per label.
In any such labeling scheme, some labels must be of length $\Omega(\sqrt{n})$.
For completeness, we will
also provide a matching upper bound of $\tilde\Oh(\sqrt{n})$.

\medskip

We start with recalling the following connection between unit-Monge matrices
and permutation matrices. $P$ is a permutation matrix if every row and every
column contains at most one 1 and 0s elsewhere. Then, it is straightforward
to verify that, for any permutation matrix $P$ the matrix $M$ defined as
$M[i,j]=\sum_{i' \geq i, j'\geq j} P[i',j']$ is a unit-Monge matrix. In fact,
essentially any unit-Monge matrix can be obtained through such transformation.
This is known, see e.g. Section 2 in~\cite{Tiskin}, but we provide a proof for completeness.

\begin{lemma}
\label{lem:unitperm}
For any unit-Monge matrix $M$, there exists a permutation matrix $P$, such that
\[
M[i,j] = H[i]+V[j]+\sum_{i' \geq i, j'\geq j} P[2i',2j'].
\]
where $H$ and $V$ are vectors of length $n$ with non-negative entries bounded
by $n$. 
\end{lemma}

\begin{proof}
We define an $n\times n$ matrix $P'$ as follows:
\[
P'[i,j] = M[i+1,j+1] + M[i,j] - M[i,j+1] - M[i+1,j].
\]
By Monge, clearly $P'[i,j] \geq 0$, and by unit $P'[i,j]\leq 2$. In fact, unit also implies that
the sum in every row and every column of $P'$ is at most 2. To see this for rows, consider
$P'[i,1]+P'[i,2]+\ldots+P'[i,n]$. After telescoping, this is $P'[i,1]-P[i+1,1]+P[i+1,n]-P[i,n]$,
so by unit at most 2 as claimed. Then, consider $\sum_{i' \geq i,j'\geq j} P'[i',j']$. After substituting
the definition of $P'$ and telescoping, this becomes $M[i,j]+M[n,n]-M[i,n]-M[n,j]$.
Hence, if we define $H[i]=M[i,n]-M[n,n]/2$ and $V[j]=M[n,j]-M[n,n]/2$ it holds that
$M[i,j] = H[i]+V[j] + \sum_{i' \geq i, j' \geq j} P'[i',j']$. Finally, we create an $2n \times 2n$
matrix $P$, where every $2\times 2$ block corresponds to a single $P'[i,j]$, that is,
the sum of values in the block is equal to $P'[i,j]$. It is always possible to define
$P$ so that it is a permutation matrix. To see this, consider a row of $P'$. The values
there sum up to at most 2, say $P'[i,j]=P'[i,j']=1$ for some $j<j'$. Then, $P'[i,j]$ should
correspond to a 1 in the first row of its block and $P'[i,j']$ to a 1 in the second row of
its block. If $j=j'$ then in the corresponding block we create two 1s, one per row.
Columns are chosen with a symmetric reasoning.
\end{proof}

Due to the above lemma, we can focus on assigning a label to every row
and column of $P$,
such that given the label of the $i$-th row and the $j$-th column we can
compute
$\sum_{i' \geq i, j'\geq j} P[i',j']$. We call this problem labeling $n\times
n$ unit-Monge matrices for dominance
sum queries.

\begin{lemma}
\label{lem:lowerboundMonge}
Labeling unit-Monge $n\times n$ matrices for dominance sum queries can be done
with $\Oh(\sqrt{n}\log n)$ bits.
\end{lemma}

\begin{proof}
We can assume that there is exactly one 1 in every row and column of $P$.
Therefore, the input is fully described
by a permutation $\pi$. Any permutation on $n$ elements can be decomposed by up
to $\sqrt{n}$ increasing
subsequences $I_{1},I_{2},\ldots$ and up to $\sqrt{n}$ decreasing subsequences
$D_{1},D_{2},\ldots$.
The label of every row and every column consists of $\Oh(\log n)$ bits stored
for every such subsequence,
thus $\Oh(\sqrt{n}\log n)$ bits in total. We think of every subsequence as a
set of points
$(x_{1},y_{1}),(x_{2},y_{2}),\ldots$ and the $\Oh(\log n)$ bits corresponding
to this subsequence
in the label of the $i$-th row and the $j$-th column should be enough to
determine
the number of points $(x_{k},y_{k})$ such that $x_{k}\geq i$ and $y_{k}\geq j$.
We separately describe
what should be stored for an increasing subsequence and then for a decreasing
subsequence.

Consider an increasing subsequence consisting of points
$(x_{1},y_{1}),(x_{2},y_{2}),\ldots$, such that
$x_{k}<x_{k+1}$ and $y_{k}<y_{k+1}$ for every $k=1,2,\ldots$. Then, the label
of the $i$-th row
stores the smallest $k$ such that $x_{k}\geq i$, and similarly the label of the
$j$-th row stores
the smallest $k$ such that $y_{k}\geq j$. By taking the maximum of these two
numbers
we can determine the number of points $(x_{k},y_{k})$ such that $x_{k}\geq i$
and $y_{k}\geq j$.

Now consider a decreasing subsequence consisting of points
$(x_{1},y_{1}),(x_{2},y_{2}),\ldots$, such that
$x_{k}<x_{k+1}$ and $y_{k}>y_{k+1}$ for every $k=1,2,\ldots$. Then, the label
of the $i$-th row
stores the smallest $k$ such that $x_{k}\geq i$. The label of the $j$-th row
stores the largest
$k$ such that $y_{k}\geq j$. Denoting the number stored for the $i$-th row and
the $j$-th row
by $\ell$ and $r$, respectively, the number of points $(x_{k},y_{k})$ such that
$x_{k}\geq i$ and
$y_{k}\geq j$ can be calculated as $\max(0,r-\ell+1)$.
\end{proof}

\begin{lemma}
\label{lem:upperboundMonge}
Labeling unit-Monge $n\times n$ matrices for dominance sum queries requires
$\Theta(\sqrt{n})$ bits.
\end{lemma}

\begin{proof}
We conceptually divide an $n\times n$ matrix $P$ into blocks of size
$\sqrt{n}\times \sqrt{n}$, thus
creating an $\sqrt{n}\times \sqrt{n}$ matrix $B$, where every entry $B[i,j]$
corresponds to a block
of $P$. For every block $B[i,j]$ we choose one bit $b[i,j]$. We will show that
then it is always possible
to construct the matrix $P$, such that all bits $b[i,j]$ can be retrieved from the
labels of rows of the
form $1+\alpha\cdot \sqrt{n}$ and columns of the form $1+\alpha\cdot \sqrt{n}$.
Then it follows
that we can encode $n$ bits of information in $2\sqrt{n}$ labels, hence one of
these labels must consist
of $\frac{1}{2}\sqrt{n}$ bits. It remains to construct $P$.

We construct $P$ incrementally.
We call a row or a column of $P$ active if there is no 1 there. We start with
an empty $P$ and keep
adding 1s there while making ensure that there is at most one 1 in every row
and column.
Given the labels of all rows $1+\alpha\cdot \sqrt{n}$ and all columns of the
form
$1+\alpha\cdot \sqrt{n}$ we can count 1s in every block of $P$. The goal is to
ensure that this
count is equal to $b[i,j]$. Assume that this is already the case for every
$b[i,j]$ such that
$i<i'$ or $i=i'$ and $j<j'$ and consider $b[i',j']$. If $b[i',j']=0$ we
continue. Otherwise, we have
to choose exactly one active row $r$ in the range $[1+i'\cdot
\sqrt{n},(i'+1)\cdot\sqrt{n}]$ and
exactly one active column $c$ in the range $[1+j'\cdot
\sqrt{n},(j'+1)\cdot\sqrt{n}]$, and
set $P[r,c]=1$, thus making both $r$ and $c$ inactive. This clearly guarantees
that there is
exactly one 1 in the corresponding block of $P$. The only problem is to
guarantee that there
is at least one active row and column in the appropriate ranges. However, we
have deactivated 
less than $i'$ rows in the range $[1+i'\cdot \sqrt{n},(i'+1)\cdot\sqrt{n}]$ so
far, and similarly less
than $j'$ columns in the range $[1+j'\cdot \sqrt{n},(j'+1)\cdot\sqrt{n}]$, so
indeed there is at
least one active row and column that we can use.
\end{proof}

\bibliographystyle{abbrv}

\end{document}